\pdfoutput=1  
\documentclass[]{article}  
\usepackage{url,float}
\usepackage{graphicx}
\usepackage{amsmath}
\usepackage{amsfonts}
\usepackage{amssymb}
\usepackage{latexsym}
\usepackage[boxruled]{algorithm2e}

\newcommand{\hide}[1]{}

\newcommand{\ABox}{
\raisebox{3pt}{\framebox[6pt]{\rule{6pt}{0pt}}}
}
\newenvironment{proof}{{\bf Proof:}}{\hfill\ABox}

\newtheorem{theorem}{{\bf Theorem}}

\newtheorem{lemma}{Lemma}

\newtheorem{conjecture}{Conjecture}
\newtheorem{question}{Question}
\newtheorem{definition}[theorem]{Definition}

\newcommand{\lemlab}[1]{\label{lemma:#1}}
\newcommand{\thmlab}[1]{\label{thm:#1}}

\newcommand{\figlab}[1]{\label{fig:#1}}
\newcommand{\seclab}[1]{\label{sec:#1}}

\newcommand{\lemref}[1]{\ref{lemma:#1}}
\newcommand{\thmref}[1]{\ref{thm:#1}}

\newcommand{\secref}[1]{\ref{sec:#1}}
\newcommand{\figref}[1]{\ref{fig:#1}}

{\makeatletter
 \gdef\xxxmark{%
   \expandafter\ifx\csname @mpargs\endcsname\relax 
     \expandafter\ifx\csname @captype\endcsname\relax 
       \marginpar{xxx}
     \else
       xxx 
     \fi
   \else
     xxx 
   \fi}
 \gdef\xxx{\@ifnextchar[\xxx@lab\xxx@nolab}
 \long\gdef\xxx@lab[#1]#2{{\bf [\xxxmark #2 ---{\sc #1}]}}
 \long\gdef\xxx@nolab#1{{\bf [\xxxmark #1]}}
 \gdef\turnoffxxx{\long\gdef\xxx@lab[##1]##2{}\long\gdef\xxx@nolab##1{}}%
}

\def\P{{\mathcal P}}
\def\r{{\rho}}

\def\S{{\Sigma}}

\def\F{{\mathcal F}}
\def\S{{\mathcal S}}
\def\T{{\mathcal T}}
\def\g{{\gamma}}
\def\l{{\lambda}}
\def\r{{\rho}}
\def\o{{\omega}}
\def\O{{\Omega}}

\def\a{{\alpha}}
\def\b{{\beta}}
\def\q{{\theta}}
\def\Th{{\Theta}}
\def\t{{\tau}}
\def\bC{{\partial C}}
\def\bG{{\partial G}}
\def\R{{\mathbb{R}}}

\newcommand{\squeezelist}{\setlength{\itemsep}{0pt}}

\title{%
Unfolding Convex Polyhedra\\
via Radially Monotone Cut Trees
} 

\author{%
Joseph O'Rourke%
    \thanks{Department of Computer Science, Smith College, Northampton, MA
      01063, USA.
      \protect\url{orourke@cs.smith.edu}.}
}

\begin{document}
\maketitle

\begin{abstract}
A notion of ``radially monotone" cut paths is introduced
as an effective choice for finding a non-overlapping edge-unfolding
of a convex polyhedron. These paths have the property that the
two sides of the cut avoid overlap locally as the cut is 
infinitesimally opened
by the curvature at the vertices along the path.
It is shown that a class of planar, triangulated convex domains always have
a radially monotone spanning forest, a forest that can be found
by an essentially greedy algorithm.
This algorithm can be mimicked in 3D and applied to polyhedra
inscribed in a sphere.
Although the algorithm does not provably find a radially monotone cut tree,
it in fact does find such a tree with high frequency,
and after cutting unfolds without overlap.
This performance of a greedy algorithm leads to the conjecture
that spherical polyhedra always have a radially monotone cut tree
and unfold without overlap.
\end{abstract}

\section{Introduction}
\seclab{Introduction}
The question of whether or not every convex polyhedron in $\R^3$ 
has a non-overlapping edge-unfolding to the plane
has been open for many years;
see~\cite{do-gfalop-07}.
Here I introduce a notion of ``radially monotone" cut paths,
and show empirically that for polyhedra inscribed in a sphere
(which I'll call \emph{spherical}\footnote{``Inscribed" is often used in the literature.}),
there exists with high frequency
a spanning tree composed of radially monotone paths,
and the corresponding unfolding avoids overlap.
A $200$-vertex example is shown in Fig.~\figref{CutsLay_s1_n200}.\footnote{
The quality of the figures had to be reduced for the arXiv.}
\begin{figure}[htbp]
\centering
\includegraphics[width=1.0\linewidth]{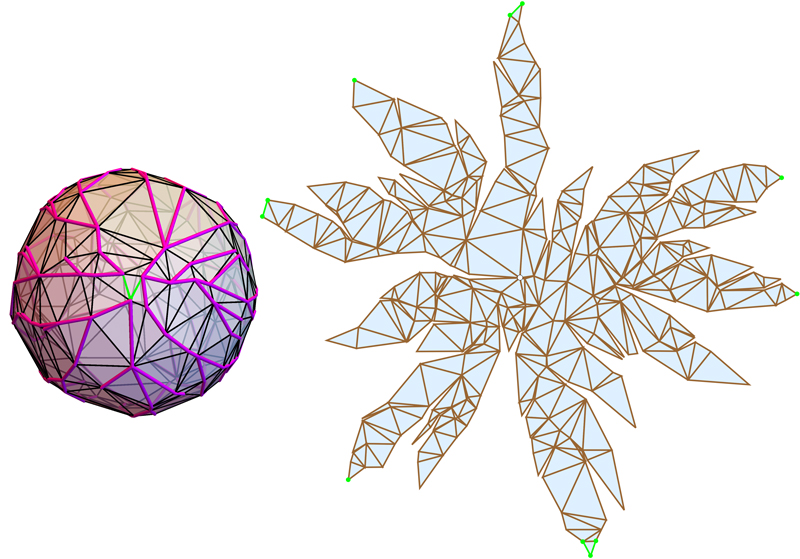}
\caption{Left: Polyhedron of $200$ vertices inscribed in sphere, with
radial monotone cut tree shown.
(The role of the green edges will be explained later.)
Right: Unfolding.
}
\figlab{CutsLay_s1_n200}
\end{figure}
This ``high frequency" claim contrasts with the near certainty of overlap
for random spanning cut trees of random spherical polyhedra,
as observed long ago: 
see Fig.~\figref{CASJORoverlap}, 
and Fig.~\figref{Lay_s12_n50_overlap} for an example of overlap.
\begin{figure}[htbp]
\centering
\includegraphics[width=0.75\linewidth]{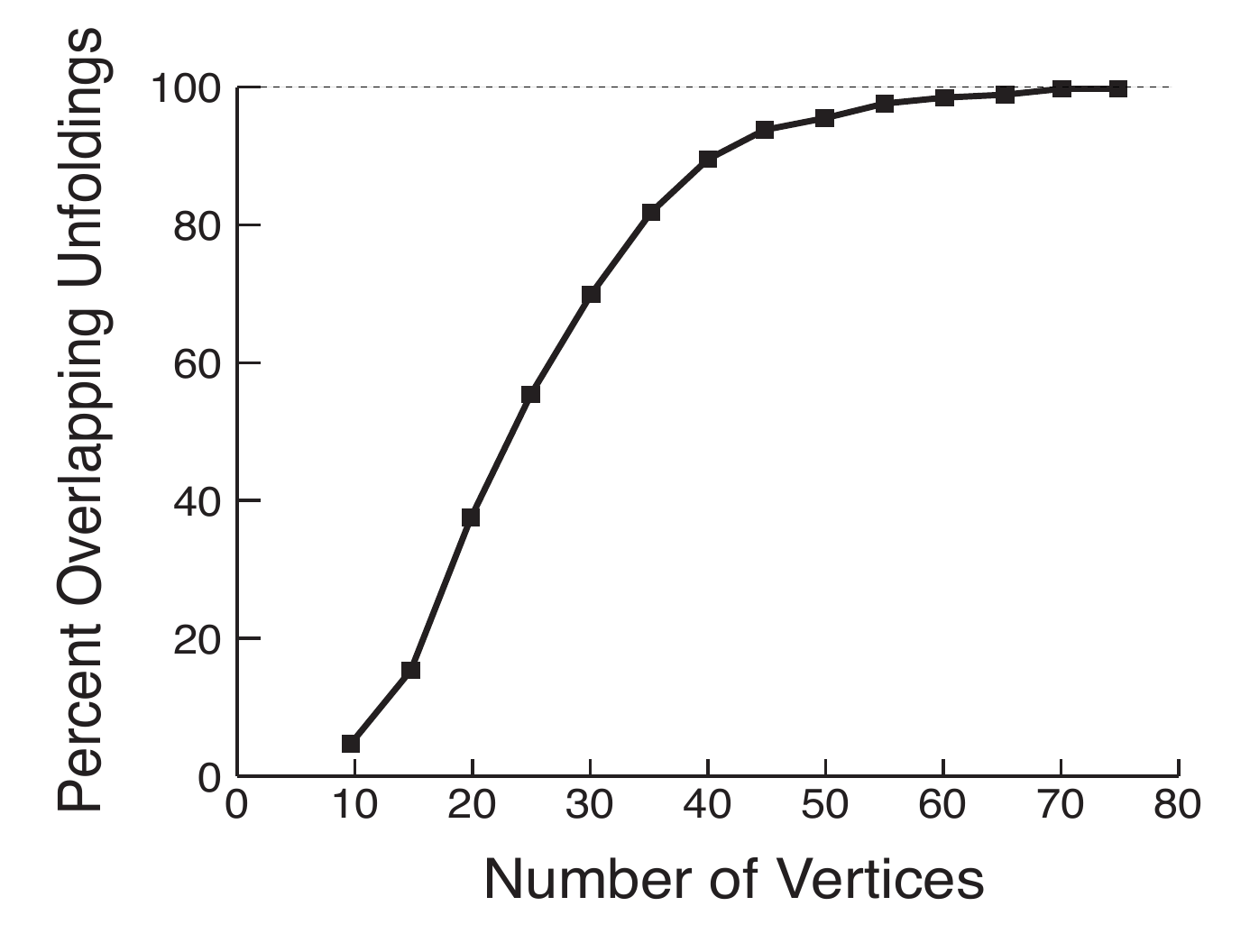}
\caption{Each point represents $5{,}000$ unfoldings of spherical polyhedra.
Fig.~22.10 in~\protect\cite[p.315]{do-gfalop-07}.}
\figlab{CASJORoverlap}
\end{figure}
\begin{figure}[htbp]
\centering
\includegraphics[width=1.0\linewidth]{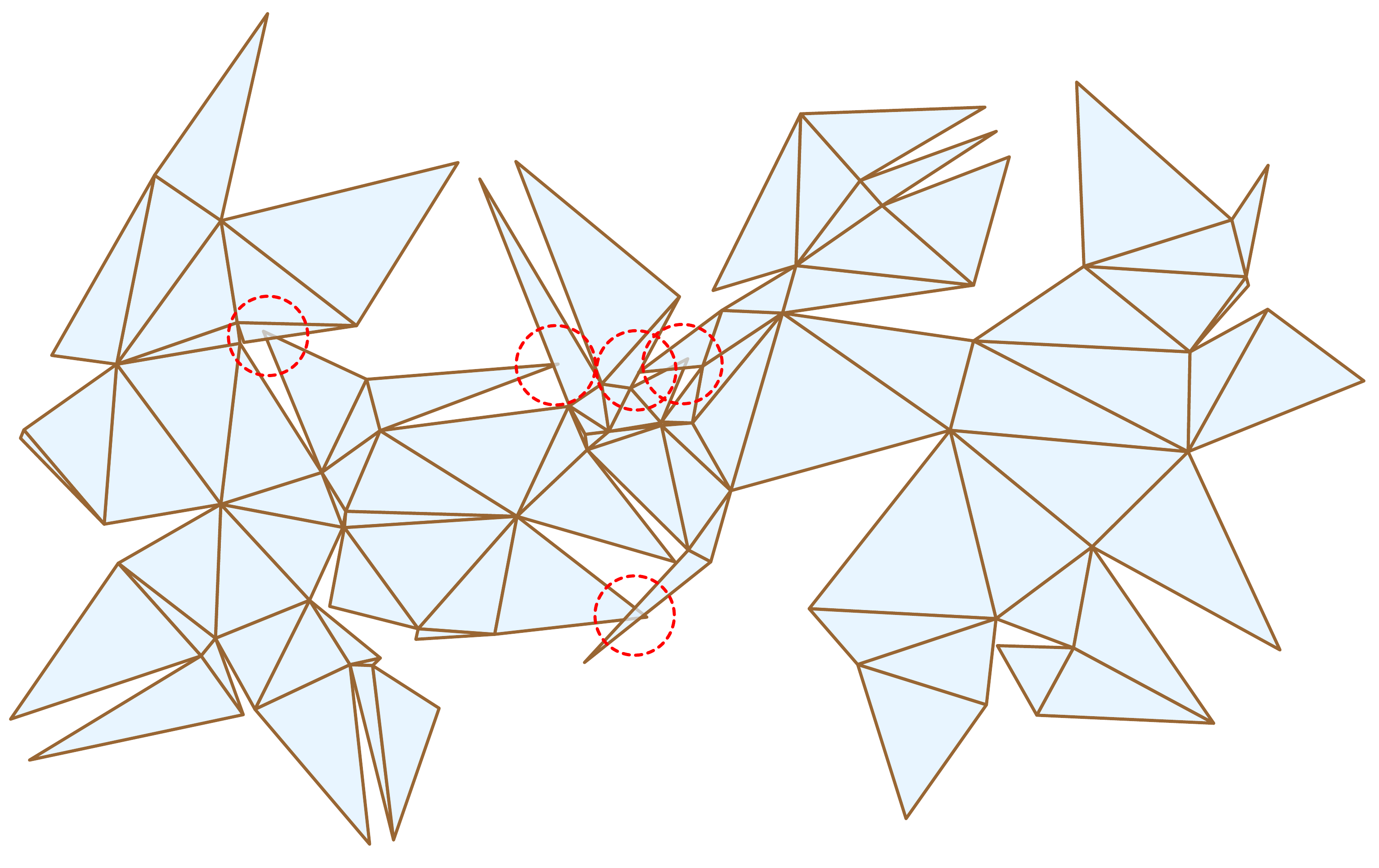}
\caption{Unfolding of a $50$-vertex spherical polyhedron with several
overlaps.}
\figlab{Lay_s12_n50_overlap}
\end{figure}

Unfortunately, the ``almost always" claim is not a theorem, but rather a
conjecture supported by data. For example, one empirical exploration run found
a radially monotone cut tree for $1{,}000$ random spherical non-obtusely
triangulated polyhedra, each of
$100$ vertices, and unfolded all $1{,}000$ without overlap.
(This claim will be hedged a bit in Section~\secref{Empirical}.)

What I have been able to prove is confined to 2D:
Theorem~\thmref{NonObConcirc} says that
``round" convex polygons, meshed with non-obtuse triangles,
always possess a radially monotone spanning forest.
The 3D empirical results use an algorithm that mimics the 2D proof,
and relies on a generalization of radial monotonicity to 3D
governed by Theorem~\thmref{LMR-nonint}.

\subsection{Definition of Radial Monotonicity in $\R^2$}
\seclab{rmDef}
We now define radial monotonicity in 2D; we will not return to 3D until
Section~\secref{rm3D}.

\subsection{Definition}
Let $C$ be a triangulated convex domain in $\R^2$,
with $\bC$ its boundary, a convex polygon. In general $C$ contains
many points in its interior that are the vertices of the triangulation.
Let $Q=(v_0,v_1,v_2,\ldots,v_k)$ be a 
simple (non-self-intersecting) directed path of edges of $C$
connecting an interior vertex $v_0$ to a boundary vertex $v_k \in \bC$.

We say that $Q=(v_0,v_1,\ldots,v_k)$ is \emph{radially monotone} (\emph{rm})\footnote{
I will also use ``rm" to abbreviate ``radially monotone"
when convenient
and unambiguous. Also, ``w.r.t." is an abbreviation for ``with respect to."}
w.r.t. $v_0$
if the distances from $v_0$ to all points of $Q$ are (non-strictly) monotonically increasing.
(Note that requiring the distance to just the vertices of $Q$ to be monotonically increasing
is not the same as requiring the distance to all points of $Q$ be monotonically increasing.)
We define path $Q$ to be \emph{radially monotone} (without qualification)
if it is radially monotone w.r.t. each of its vertices: $v_0,v_1, \ldots, v_{k-1}$.
Before exploring this definition further, we discuss its intuitive motivation.

\subsection{Motivating Intuition}
A radial monotone path $Q$ w.r.t. $v_0$ has the property that
rigidly rotating $Q$ and all its incident triangles about $v_0$
by a small angle
avoids proper overlap between the triangles to the left and to the right of $Q$.
One can imagine one triangle incident to $v_0$ reducing its angle
at $v_0$ infinitesimally, as illustrated in 
Fig.~\figref{isopPath}(a).
\begin{figure}[htbp]
\centering
\includegraphics[width=1.0\linewidth]{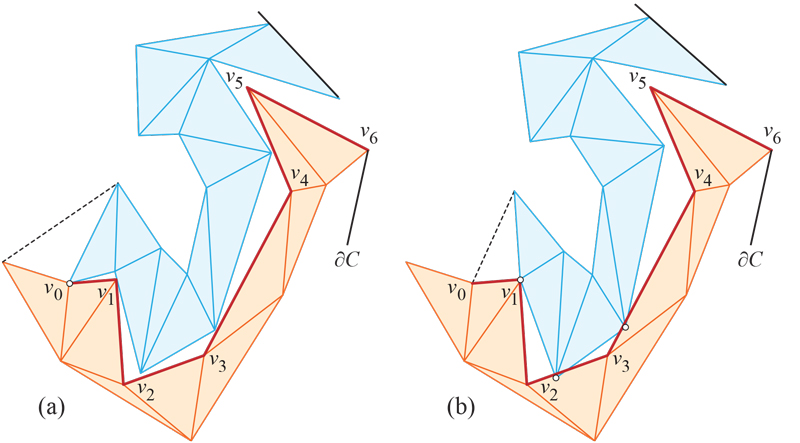}
\caption{Path $Q=(v_0,v_1,\ldots,v_6)$ is radially monotone w.r.t. $v_0$,
but not w.r.t. $v_1$.}
\figlab{isopPath}
\end{figure}
Fig.~\figref{isopPath}(b) shows that the path illustrated is not
radially monotone without qualification, because it is not w.r.t. $v_1$.

The motivation behind the definition of radial monotonicity is as follows.
Ultimately the path $Q$ will be a path of edges on a convex polyhedron $\P$ in $\R^3$.
At each vertex $v_i \in \P$, there will be some positive curvature $\o_i>0$, which
represents the ``angle gap" when the neighborhood of $v_i$ is flattened to the plane.
We can view $\o_i$ as separating the left-half of the cut $Q$ from the right-half at $v_i$.

As a consequence, if a path $Q$ is rm, then ``opening" the path
with sufficiently small curvatures $\o_i$ at each $v_i$ will avoid
overlap between the two halves of the cut path.
Whereas if a path is not rm, then there is some opening 
curvature assignments $\o_i$ to the $v_i$ that would cause overlap:
assign a positive curvature $\o_j>0$ to the first vertex $v_j$ at which
radial monotonicity is violated, and assign the other vertices zero or negligible curvatures.
Thus rm cut paths are locally (infinitesimally) ``safe," and non-rm paths are potentially overlapping.
This potential overlap may not be realizable, for the $\o_i$ cannot be assigned
arbitrarily, but must derive from a convex polyhedron.
So the guarantee is one-way: 

\fbox{
Radial monotonicty $\implies$ safe (non-overlapping) infinitesimal opening.
}

\subsection{Trees and Spanning Forests}
Continuing to concentrate on a planar triangulated convex domain $C$,
we extend the notion of radial monotonicity to trees and forests
in the natural manner.
A tree $T$
rooted on a boundary vertex $v_k \in \bC$ and containing no other boundary vertex
is radially monotone if the path from every leaf of $T$ to $v_k$ is rm.
A radially monotone spanning forest for $C$ is a collection of
boundary-rooted rm trees that span the interior vertices of $C$.

\subsection{Spherical Caps}
To presage why we concentrate on convex domains,
we turn briefly again to 3D.
Fig.~\figref{CutsLay_s2_pi3_n500} shows a convex cap composed of
the faces of a spherical polyhedron whose normals are within $60^\circ$ of the
vertical. 
We will unfold convex caps via radially monotone cut forests.
And ultimately we will show how to view a complete spherical polyhedron
as a convex cap.
\begin{figure}[htbp]
\centering
\includegraphics[width=1.0\linewidth]{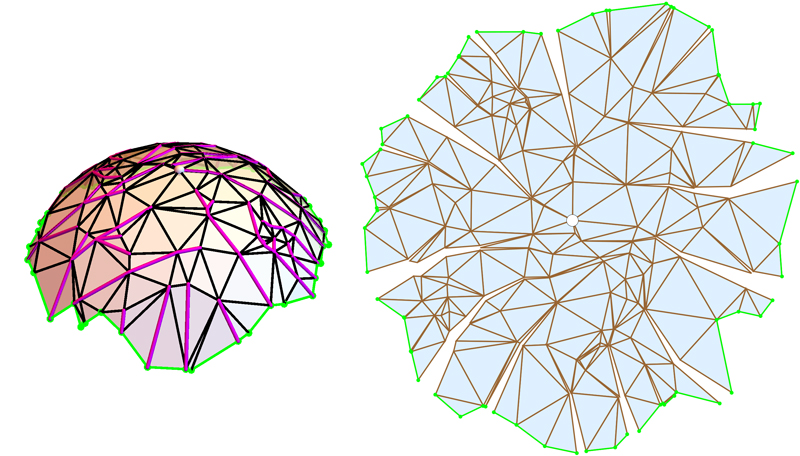}
\caption{Left: Spherical cap of $147$ vertices, with
radial monotone cut forest marked.
Right: Unfolding. Green edges are part of $\bC$.
}
\figlab{CutsLay_s2_pi3_n500}
\end{figure}

\section{Properties of Radially Monotone Paths}
\seclab{PropsRM}
We now embark on a rather lengthy description of properties of rm paths in 2D.

\subsection{Radial Circles}
The condition for $Q$ to be rm w.r.t. $v_0$ can be
interpreted as requiring $Q$ to cross every circle centered on $v_0$ at
most once; see Fig.~\figref{RadialCircles}.
The concentric circles viewpoint makes it evident
that infinitesimal rigid rotation of $Q$ about $v_0$ to $Q'$ ensures that
$Q \cap Q' = \{ v_0 \}$.
\begin{figure}[htbp]
\centering
\includegraphics[width=\linewidth]{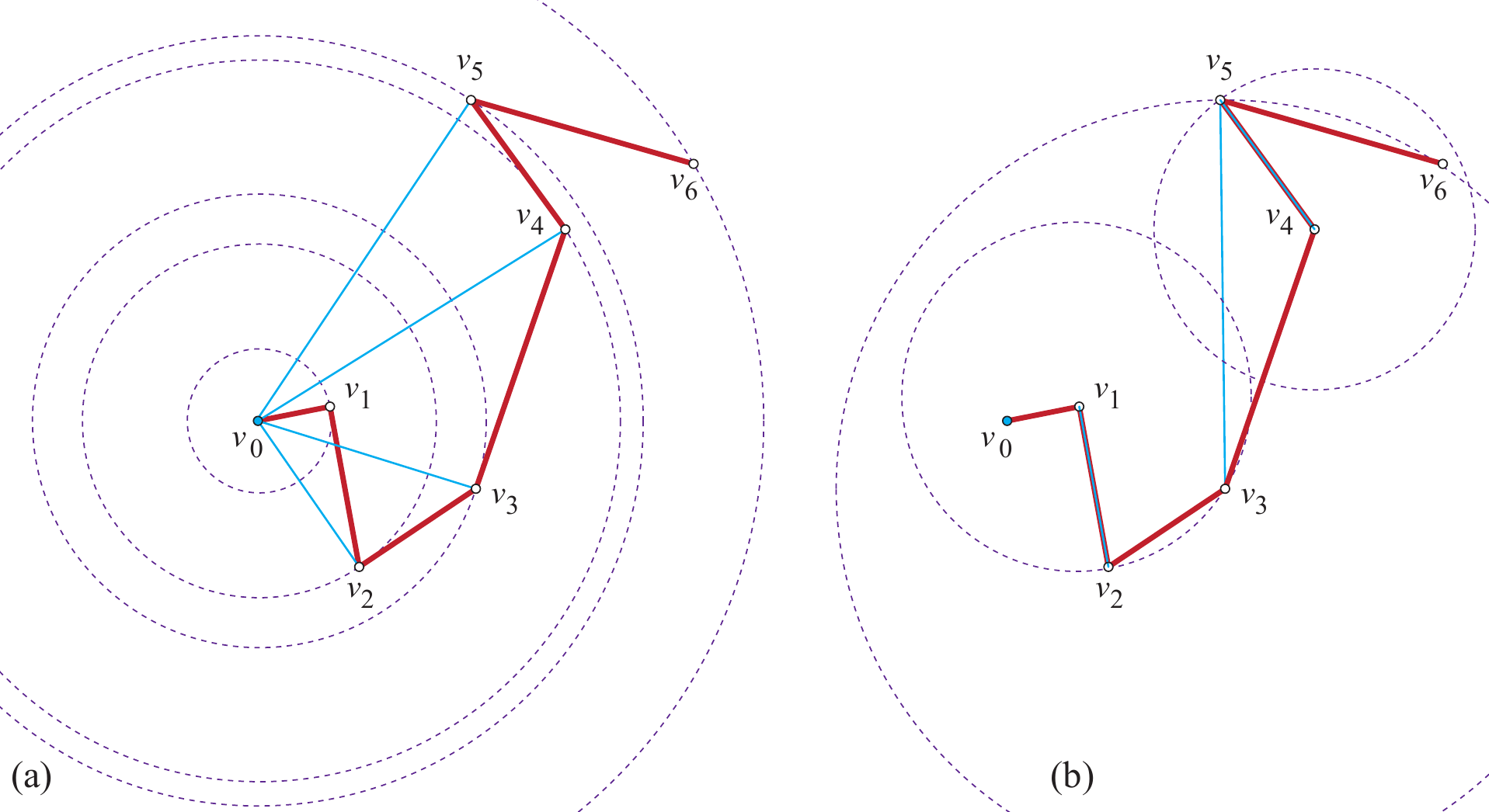}
\caption{(a)~A chain radially monotone w.r.t. $v_0$.
(b)~The chain is not radially monotone w.r.t. $v_1$
(violation at $v_2$),
or w.r.t. $v_3$
(violation at $v_5$),
or w.r.t. $v_4$
(violation at $v_5$).}
\figlab{RadialCircles}
\end{figure}

An equivalent definition is as follows.
Let $\a_j(v_i) = \angle (x, v_j, v_{j+1})$.
Then $Q$ is rm w.r.t. $v_i$ if
$\a_j(v_i) \ge \pi/2$ for all $j>i$.
For if $\a_j(v_i) < \pi/2$, $Q$ violates monotonicity at $v_j$,
and if $\a_j(v_i) \ge \pi/2$, then points along the segment $(v_j, v_{j+1})$
increase in distance from $v_i$.

\subsection{Non-Properties of rm Paths}
Let the \emph{turn angle} $\t_i$ of path $Q$ at $v_i$ be the
signed angle between the vectors $v_i - v_{i-1}$ and $v_{i+1}-v_i$;
so $\t_i \in [-\pi,\pi]$, with $\t_i=0$ meaning that
the joint at $v_i$ is straightened.

It should be clear that
no turn angle in a radially monotone path can exceed $\pi/2$:
$\t_i \le \pi/2$ for all $i=1,\ldots,k{-}1$.
Although this condition is necessary, it is not sufficient:
if $Q=(v_0,v_1,v_2,v_3)$ has unit link lengths, and turn angles $\t_1=\t_2=\pi/2$, 
forming a $\sqcap$-shape,
then $Q$ is not rm w.r.t. $v_0$, violating rm at $v_2$.

Note the definition of radial monotonicity considers paths directed from 
$v_0$ to $v_k$.
A path radially monotone in one direction need not be rm when reversed:
see Fig.~\figref{RMrev}.
\begin{figure}[htbp]
\centering
\includegraphics[width=0.5\linewidth]{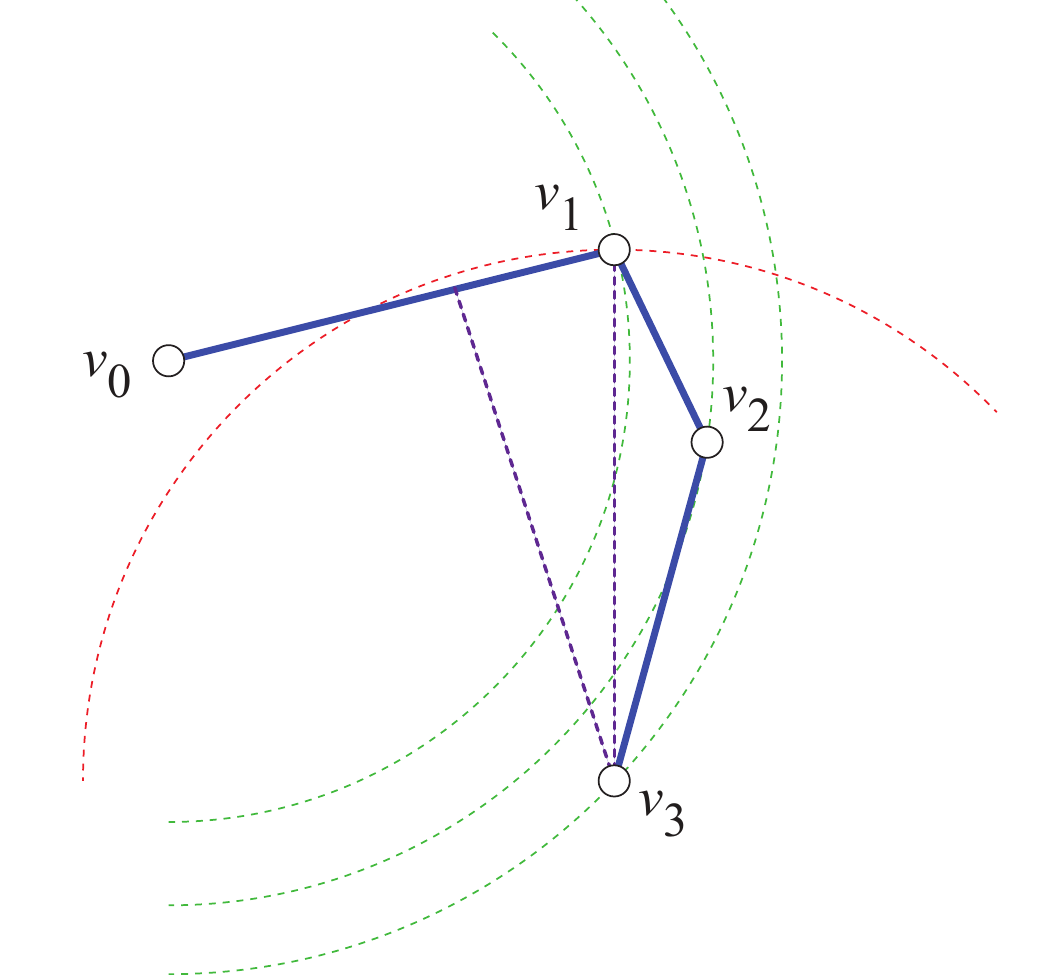}
\caption{$Q=(v_0,v_1,v_2,v_3)$ is radially monotone, 
but its reverse, $(v_3,v_2,v_1,v_0)$, is not rm.
Shortcutting the former rm path to $(v_0,v_1,v_3)$ destroys radial monotonicity.}
\figlab{RMrev}
\end{figure}

Another property one might hope holds is that shortcutting a rm path
$Q=(v_0,\ldots,v_{i-1},v_i,v_{i+1},\ldots,v_k)$ to 
$Q'=(v_0,\ldots,v_{i-1},v_{i+1},\ldots,v_k)$
would retain radial monotonicity, but this is false in general, as also illustrated in Fig.~\figref{RMrev}.

\subsection{Searching for rm Paths}
We will search for rm paths to form a rm cut forest to unfold a spherical cap.
The search will be incremental, growing existing paths.
So we next turn to properties that will allow a partial rm path to be extended,
on either end.
Again we concentrate on a planar convex domain $C$.

There are two basic strategies: Start from an interior vertex of $C$, and
grow toward $\bC$; or start from a boundary vertex, and grow toward the interior.
The former is attractive because, as Lemma~\lemref{ThetaBounds}
will show, a partial path can always be grown ``forward."
Meanwhile, ``backward" growth is not always possible,
as will be shown in Section~\secref{BackwardGrowth}.
However, so far I have only been able to prove that a
rm spanning forest exists by growing from the boundary inward.

\subsection{Forward Extension: O-cone $\Th_j$}
\seclab{o-cone}
For path $Q=(v_0,v_1,\ldots,v_k)$,
let $\Th_j$ be the range of angles within which $(v_{j+1}-v_j)$ must
lie for $Q$ to be rm.
We call $\Th_j$ the \emph{o-cone}  at $v_j$,
because its bounding rays are orthogonal to the \emph{cone} at $v_j$:
the smallest cone with apex at $v_j$ within which all of $Q$ lies.
See Fig.~\figref{oCones6}.

For segment $(v_j,v_{j+1})$ to be rm w.r.t. $v_i$, that segment must fall outside
the circle centered on $v_i$ of radius $|v_j-v_i|$. Thus each $v_i$
contributes a halfplane constraint at $v_j$, and it is the intersection of
those halfplanes that determine the o-cone $\Th_j$.
And clearly the extremes of the o-cone are orthogonal to the extremes of the cone.

\begin{figure}[htbp]
\centering
\includegraphics[width=0.75\linewidth]{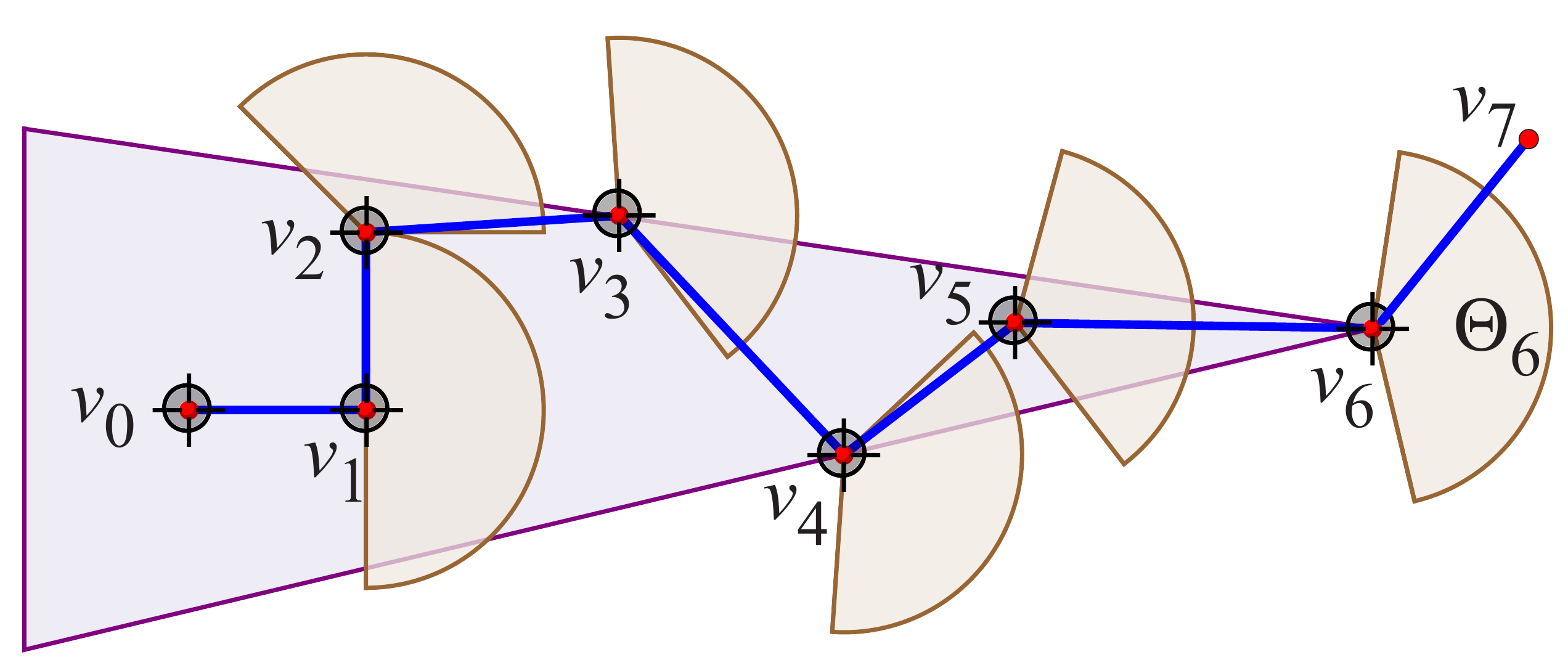}
\caption{The o-cones $\Th_j$ of a radially monotone path.
The cone at $v_6$ (illustrated) passes through $v_3$ and $v_4$.}
\figlab{oCones6}
\end{figure}

Let $|\Th_j|$ be the measure of the o-cone angle.
The following lemma is not used in the sequel, and is included just for intuition.
\begin{lemma}
For any radially monotone path $Q=(v_0,\ldots,v_k)$, $|\Th_j| \in [\pi/2,\pi]$ for all $j$.
\lemlab{ThetaBounds}
\end{lemma}
\begin{proof}
Initially $|\Th_1| =\pi$, and assume the claimed bound holds
up to $v_k$. Now consider extending $Q$ with $v_{k+1}$.
The cone for $Q$ encompasses the convex hull of $Q$, as illustrated in
Fig.~\figref{ThetaBounds}.
Let $\a_k$ be the angle of the cone. Then the angle of the o-cone is $\pi-\a_k$.
If $v_{k+1}$ is close to $v_k$ on the upper or lower boundary of $\Th_k$, 
then $|\Th_{k+1}|$ approaches $\pi/2$ from above as the distance
$|v_k-v_{k+1}| \to 0$.
If $v_{k+1}$ lies on the line containing $(v_{k-1}, v_k)$, and so extends the previous
link collinearly, then $|\Th_{k+1}|$ approaches $|\Th_k|$ as $|v_k-v_{k+1}| \to 0$.
There is a continuous variation in $|\Th_{k+1}|$ between these extremes for
$v_{k+1}$ arbitrarily close to $v_k$.

Now consider moving $v_{k+1}$ from near $v_k$ along a line that falls within $|\Th_k|$.
The cone apex angle $\a_{k+1}$ narrows monotonically, and so $|\Th_{k+1}|$ monotonically
increases. So $|\Th_{k+1}|$ is always $\ge \pi/2$.
And clearly $\a_{k+1} \ge 0$ and so $|\Th_{k+1}| \le \pi$.
\end{proof}
\begin{figure}[htbp]
\centering
\includegraphics[width=1.0\linewidth]{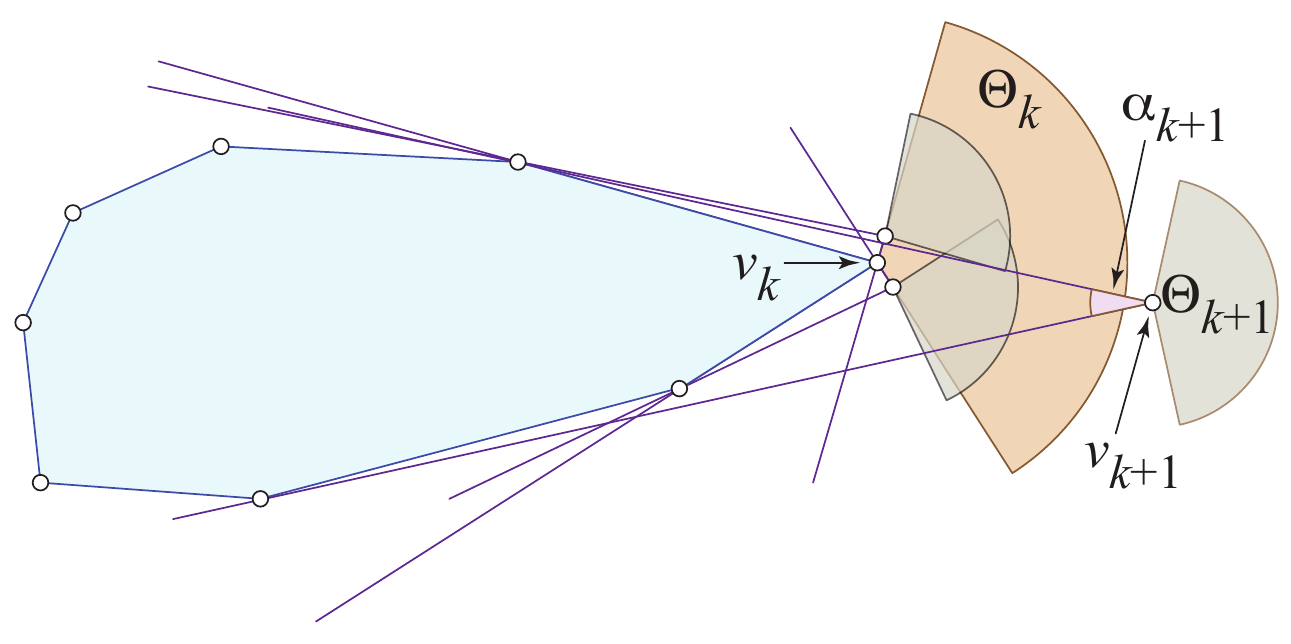}
\caption{$\Th_{n+1} = \pi - \a_{n+1}$ varies from $\pi/2$ near $v_n$, to approaching $\pi$ distant from $v_n$.}
\figlab{ThetaBounds}
\end{figure}

\subsection{Backward Growth}
\seclab{BackwardGrowth}
Consider now a rm path $Q=(v_0,v_1,v_2,\ldots,v_k)$ that we would like
to extend ``backwards," with $v'$ prior to $v_0$, while retaining
rm for $Q'=(v',v_0,v_1,\ldots,v_k)$.
Define $R=R(Q)$ to be the region of the plane within which $v'$ can lie
while retaining $Q'$ rm.
Each edge $(v_i,v_{i+1})$ of $Q$ contributes a halfplane constraint to $R$:
$v'$ must lie in the halfplane whose boundary is orthogonal to 
$(v_i,v_{i+1})$, passes through $v_i$, and excludes $v_{i+1}$.
For otherwise, $\angle v', v_i,v_{i+1} < \pi/2$, and rm w.r.t. $v'$
would be violated at $v_i$.
$R$ is then the intersection of all these halfplanes;
see Fig.~\figref{RMNegCone}.
\begin{figure}[htbp]
\centering
\includegraphics[width=0.75\linewidth]{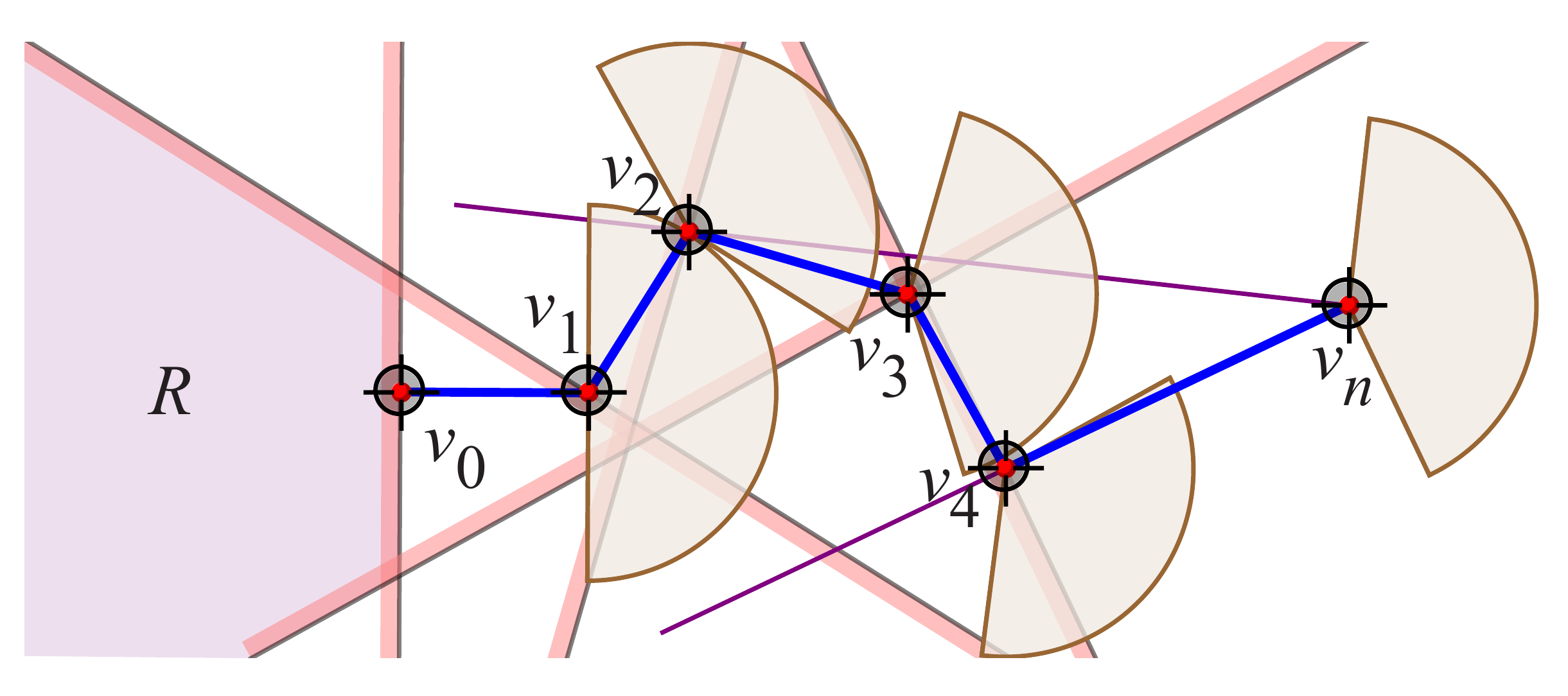}
\caption{$R$ is determined here by the halfplanes generated
by $(v_0,v_1)$, $(v_1,v_2)$, and $(v_3,v_4)$.}
\figlab{RMNegCone}
\end{figure}

It is possible that $Q$ cannot be extended backwards at all: $R$ could equal $\{v_0\}$,
as illustrated in 
Fig.~\figref{RMEmptyRootExtension}.
\begin{figure}[htbp]
\centering
\includegraphics[width=0.75\linewidth]{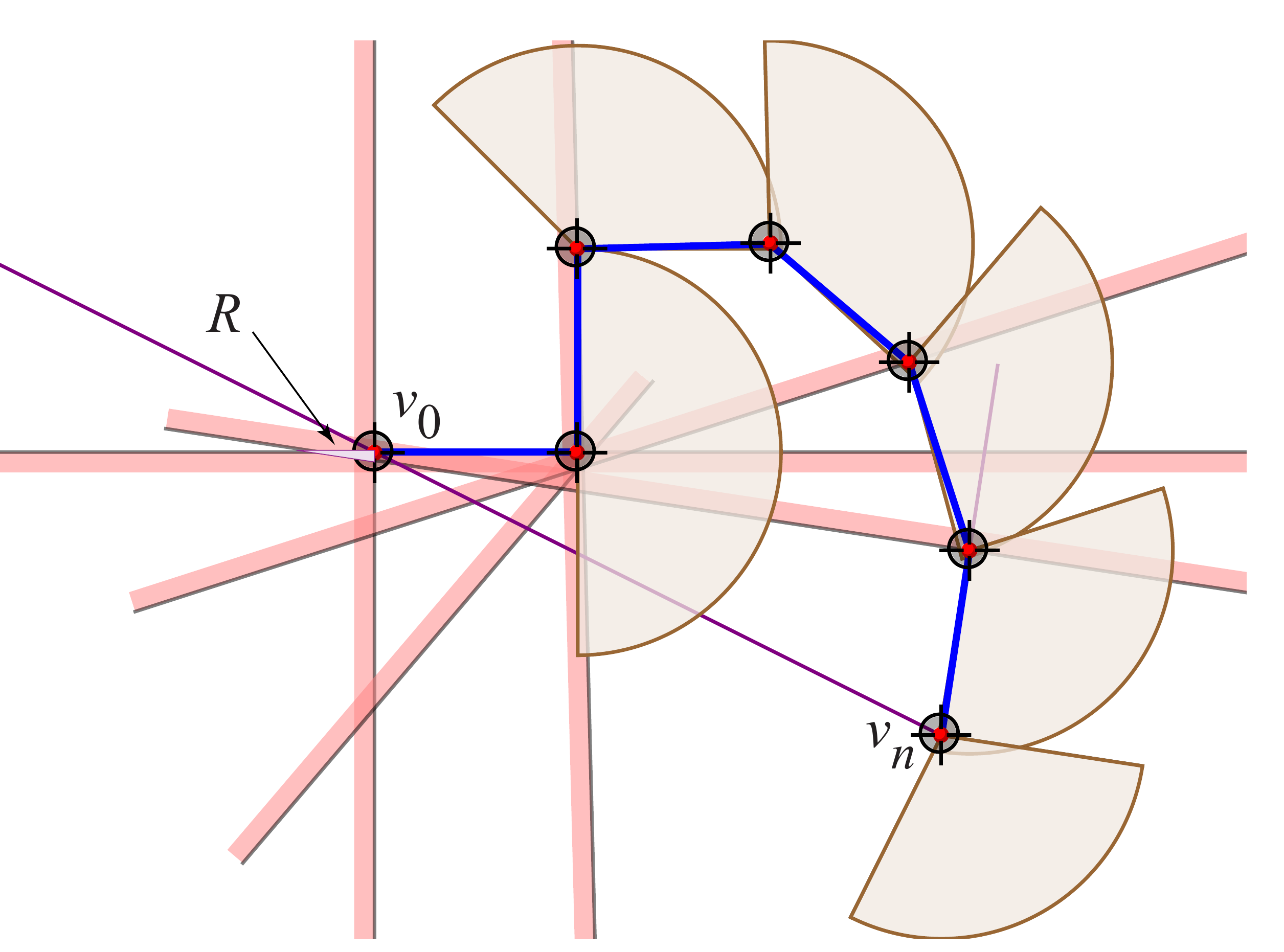}
\caption{$R$ is the small triangle with one corner $v_0$.}
\figlab{RMEmptyRootExtension}
\end{figure}

\subsection{Logarithmic Spirals}
\seclab{spirals}
Although there is a sense in which radially monotone paths are somewhat 
``straight," in that they increase in distance from $v_0$ and never turn too
sharply, the example in Fig.~\figref{RMEmptyRootExtension} shows that
intuition cannot be pushed too far. In this section we identify the
extreme rm path, which turns out to be an approximately $75^\circ$ logarithmic spiral.

A logarithmic spiral can be expressed in polar coordinates by the equation $r= a e^{b \q}$.
For our purposes we can take the scale factor $a$ to be $1$.
$b=1/\tan \phi$, where $\phi$ is the constant angle between
the radial vector to a point $p$ on the curve, and the tangent of the curve at $p$.
This constant angle is the determining characteristic of such spirals.
Fig.~\figref{Spiral80deg} shows a typical spiral. When $\phi = \pi/2$,
the spiral becomes a circle. As $\q \to -\infty$, $r \to 0$ and the spiral approaches
the origin.
\begin{figure}[htbp]
\centering
\includegraphics[width=0.6\linewidth]{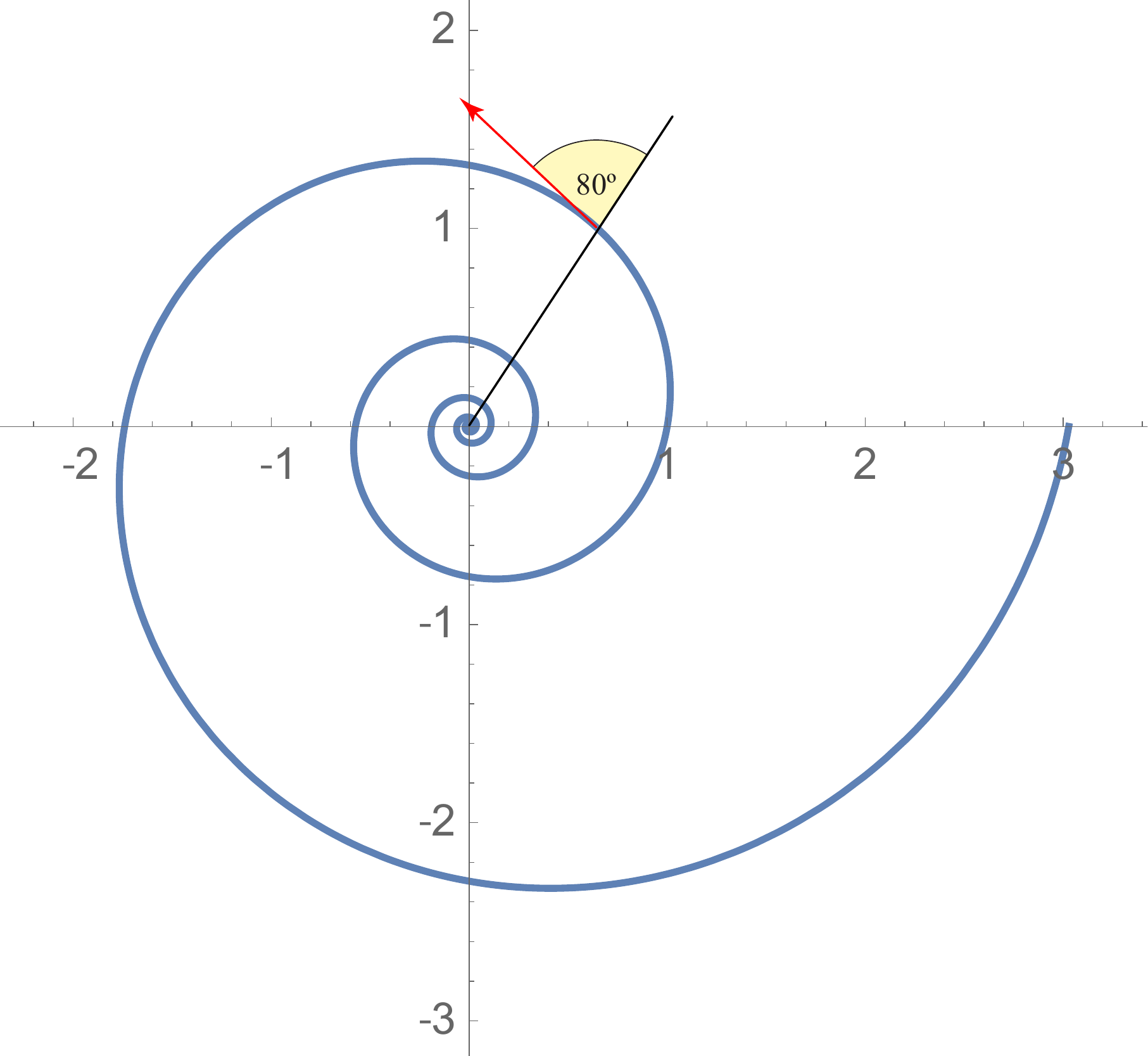}
\caption{A logarithmic spiral with $\phi=80^\circ$. Here $\q$ is plotted out to $2 \pi$.}
\figlab{Spiral80deg}
\end{figure}

We can extend the definition of radial monotonicity to apply to smooth curves:
a directed curve is radially monotone if it is rm w.r.t. to every point $p$ on the curve, i.e.,
the distance from $p$ to points beyond $p$ on the curve (non-strictly) increases
monotonically.
It turns out that the spiral in Fig.~\figref{Spiral80deg} is not radially monotone,
but a numerical calculation shows that log spirals with 
$\phi \le 74.655^\circ$ are radially monotone.
Fig.~\figref{LogSpiralExtreme} shows the extreme spiral,
what I will call a $75^\circ$-spiral.
\begin{figure}[htbp]
\centering
\includegraphics[width=0.95\linewidth]{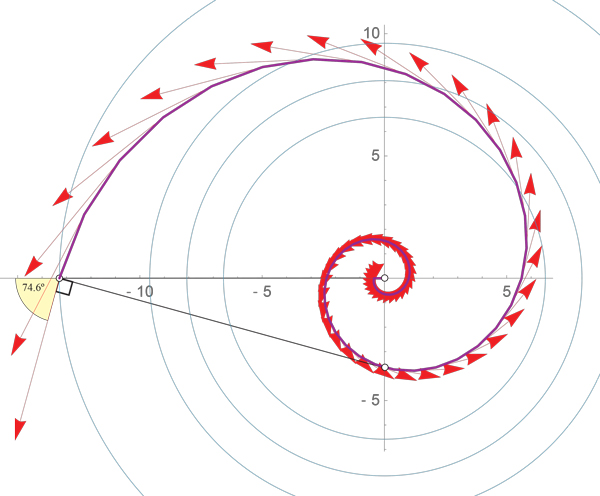}
\caption{$74.655^\circ$ spiral. The tangent at any point 
derived from polar angle $\q$ on the curve
makes a $90^\circ$ angle w.r.t. the point at polar angle $\q-3 \pi/2$.}
\figlab{LogSpiralExtreme}
\end{figure}

\section{Radially Monotone Forests Exist for Non-Obtuse Triangulations}
\seclab{NonObTri}
In this section we derive our main 2D result, which is the inspiration for
the 3D algorithm to follow.

\subsection{Setting}
Again we concentrate on be a triangulated domain $C$ whose boundary $\bC$ is a convex polygon.
Later will impose some conditions on the shape of $C$ and on its triangulation.
Our goal is to find a radially monotone spanning forest $\F$ for $C$.
$\F$ is a collection of disjoint trees, each with exactly one vertex on $\bC$ (its root),
following edges of the triangulation, and spanning all interior vertices (those not on $\bC$).
Note that each interior vertex has a unique path to $\bC$, because each interior
vertex lies in a unique tree rooted on $\bC$.

As mentioned previously, we will eventually turn to 3D when $C$ will be a spherical cap, and cutting $\F$
will permit $C$ to be flattened into the plane as a single piece.
Our goal then will be to avoid overlap in this edge-unfolding.
and a rm-forest guarantees an
infinitesimally safe opening.

In 2D, however, $C$ is already flat and
cutting $\F$ just has the property that the dual graph of
$C$'s triangulation---with each
triangle a node and two triangles adjacent across an uncut edge---is a tree,
and so a single piece.
The reason that cutting $\F$ leaves $C$ in one piece is the restriction that
each tree only touch the boundary at one vertex. 
For if a tree touched two boundary vertices,
then it would form a cycle with the portion of $\bC$ between those vertices,
disconnecting a piece from the remainder.
We insist that $\bC$ be convex because one can rather easily force overlap
in the 3D situation by designing a nonconvex boundary appropriately.
(However, as is clear from Fig.~\figref{CutsLay_s2_pi3_n500},
a weaker condition than convexity may suffice.)


Before commencing, it is important to remark that
not every triangulation with a convex boundary has an rm-forest. 
See Appendix~1 for a counterexample.
Thus we do need to seek conditions that guarantee existence.

\subsection{Overview of Algorithm}
Although one could imagine a search over all possible spanning forests,
instead we will implement a tightly regimented
greedy search.
First, $C$ is placed in a minimum enclosing circle
with center $x$.
Second, the internal vertices are processed in order of their distance to $x$,
with those closest to $\bC$ first, and the closest vertex to $x$ last.
So the rm trees are grown inward toward $x$.
Third, we grow \emph{hourglass paths}, which are a restricted
class of radially monotone paths.
Before describing the details, we sketch the algorithm at a high level 
in the boxed pseudocode below.

\begin{algorithm}[htbp]
\caption{Find rm cut forest $\F$ for planar $C$}
\DontPrintSemicolon

    \SetKwInOut{Input}{Input}
    \SetKwInOut{Output}{Output}

    \Input{Convex triangulated domain $C$, with bounding circle center $x$}
    \Output{Radially monotone cut forest $\F$}
    
    \BlankLine
    \tcp{$\F$ is grown from $\bC$ inward}
     \BlankLine
    
    $\F \leftarrow \varnothing$
    
    Sort interior vertices by distance from center $x$, those nearest $\bC$ first.\;
     
    \BlankLine
     
     \tcp{Grow $\F$:}

     \ForEach{vertex $v_0$ in sorted order}{
   
     \ForEach{vertex $v_1$ already in $\F$ or on $\bC$}{
     
     Let $e=(v_0,v_1)$, if $e$ is a triangulation edge.\;
     Check if the path from $v_0$ to $\bC$ in $\F+e$ is radially monotone.\;
     If so, record its worst turnangle $\t$ (with $\t>90^\circ$ not rm). \;
    
      }
       
      Choose the $e^*$ that has the best (minimum) $\t$. \;
      
      $\F \leftarrow \F + e^*$
            
     }
     
\end{algorithm}

The order of growth of the trees in the forest $\F$
is determined by the concentric-circle sorting w.r.t. $x$.
The algorithm is greedy in the sense that among the options 
(incident triangulation edges) when
connecting $v_0$ to some further-away $v_1 \in \F$,
the ``best connection" is selected.
For a path of length at least $2$, best is defined
as the smallest worst turnangle $\t_i$,
the angle from the vector $v_i-v_0$ to the vector $v_{i+1}-v_i$.
When $\t_i > \pi/2$, the path is not rm.
Smaller $\t_i$ means straighter paths, so this is a natural choice.

We need to define what the turnangle means when the path is
of length $1$, $v_0 v_1$ with $v_1 \in \bC$.
We use the turnangle from the vector $v_1-v_0$ to the tangent
of the circle centered at $x$ through $v_1$.
The reason for this choice (which ignores the orientation of $\bC$ at $v_1$)
will be made clear below.

Note that, when a connection from $v_0$ to $v_1$ is explored for radial monotonicity, 
it is only necessary to check for rm w.r.t. $v_0$, because by construction
the path is already known to be rm w.r.t. $v_i$, $i>0$, because it was
earlier added to $\F$.

\subsection{Hourglass Paths for Halfplanes}
\seclab{HourGlassHp}
To provide intuition for hourglass paths, we first describe them 
for $C$ a halfplane, bounded on the right by a vertical line $L=\bC$.
Imagine the halfplane meshed with a triangulation.

It will be occasionally more convenient to label the vertices of a path to increase
from $\bC$ inward rather than the reverse.
We will use $u_i$ whenever indexing inward.
So $Q=(v_0,v_1,\ldots,v_k) = (u_k,u_{k-1},\ldots,u_0)$ with $u_0 \in \bC$.
So our rm paths will grow from $u_0$ inward to $u_k$
(but the direction for radial monotonicty will still be interior-to-boudary, $v_0$ to $v_k$).

We define an \emph{hourglass} $H$ as a double cone bounded by two lines
meeting at right angles. The \emph{baseline} of an hourglass is the line
passing through the apexes of the cones so that the bounding rays
make angles of $\pi/4$ with the baseline. In our halfplane example, the
baselines are all parallel to $L$.

We will now imagine growing a path $Q$ from $u_0 \in L=\bC$ inward.
We place an hourglass $H$ centered on each $u_i$,
and call the cone of $H$ pointing right, toward $u_{i-1}$ and the halfplane boundary,
the \emph{out-cone}, and the cone of $H$ pointing left, toward $u_{i+1}$ and the interior,
the \emph{in-cone} of $H$.
Finally, define an \emph{hourglass path} $Q=(u_0,\ldots,u_k)$
as one for which each edge $u_i u_{i+1}$ falls inside the out-cone of
the hourglass at $u_i$ and the in-cone of the hourglass at $u_{i+1}$.
See Fig.~\figref{HgHalfplane}.
\begin{figure}[htbp]
\centering
\includegraphics[width=0.90\linewidth]{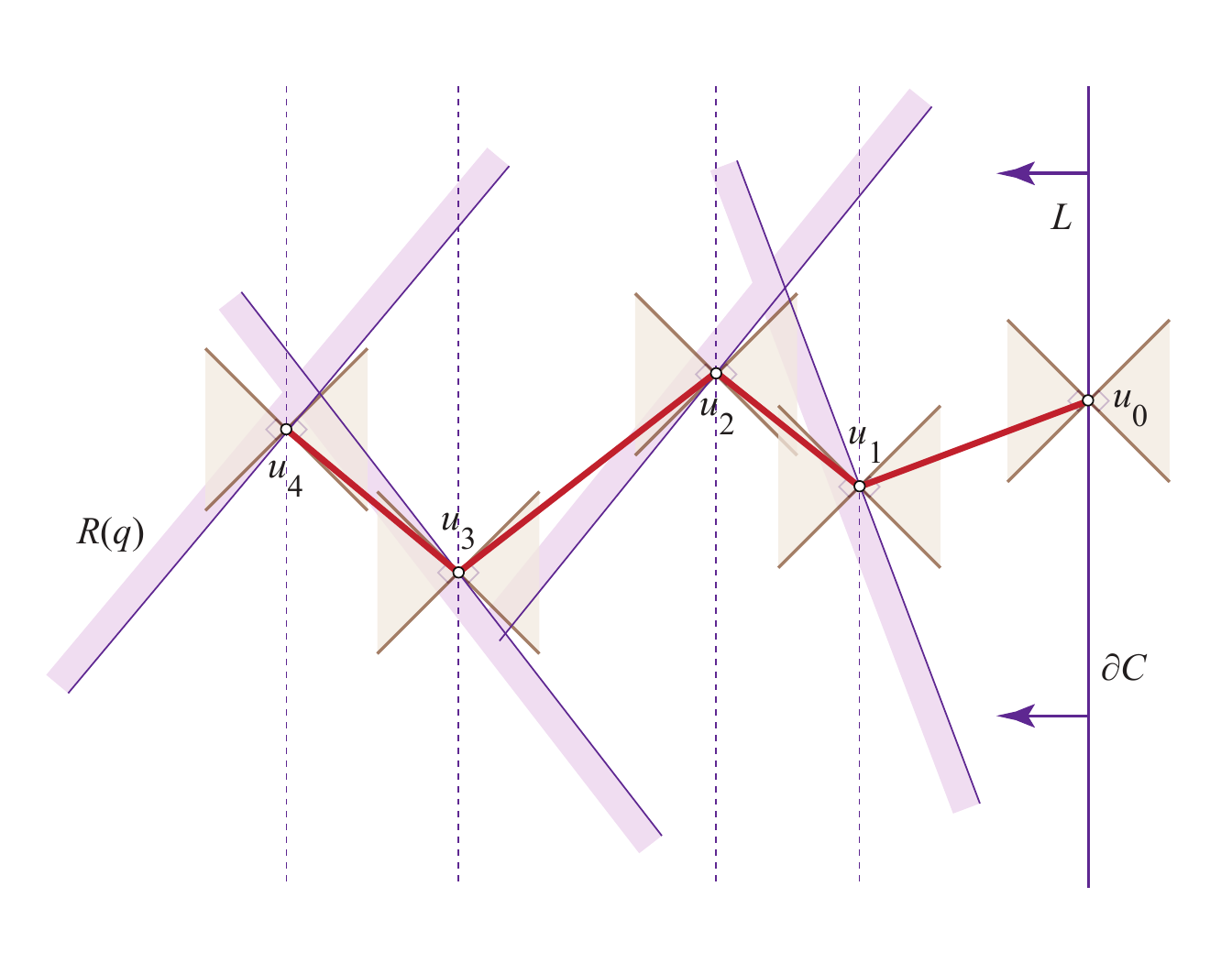}
\caption{The out-cone of the hourglass at $u_4$ falls inside the intersection
of the orthogonal halfplanes (purple) through $u_1,\ldots, u_4$.}
\figlab{HgHalfplane}
\end{figure}
Note that, if $u_i u_{i+1}$ falls inside the out-cone of
$u_i$ , then it necessarily falls within the in-cone of $u_{i+1}$ .
We retain the definition as stated because this property relies on the
parallel baselines of the hourglasses, which will not hold in more general situations.

Recall from Section~\secref{BackwardGrowth} that we can extend $Q$ 
``backwards" (inward) to $Q=(u_0,\ldots,u_k,u_{k+1})$ only if
the edge $u_k,u_{k+1}$ falls within $R(Q)$, the region formed by halfplane
intersections (and which could be a single point).
The key point is that hourglass paths can always extend:

\begin{lemma}
If $Q=(u_0,\ldots,u_k)$ is an hourglass path in $C$ a halfplane, with $u_0 \in L$,
then (a)~$Q$ is a radially monotone path, and
(b)~$R(Q)$ contains the out-cone at $u_k$.
\lemlab{HgOutCone}
\end{lemma}
\begin{proof}
The situation is as illustrated in Fig.~\figref{HgHalfplane}.
(a)
Note that $Q$'s direction for the purposes of radial monotonicity
is from $u_k$ to $u_0$. Because each hourglass cone is angled $45^\circ$
w.r.t. the vertical, it must be that the whole path $Q$ falls within the
out-cone at $u_k$. Then consider the angle $\angle u_k, u_i, u_{i-1}$
for any $u_i$. This angle must be at least $90^\circ$, exactly $90^\circ$
if every edge of $Q$ falls on (say) the lower cone boundaries.
Since the argument holds for any pair of vertices of $Q$, $Q$ is indeed rm.

(b)
$R(Q)$ is the intersection of halfplanes orthogonal to $u_{i-1} u_i$ and through $u_i$.
Falling within the hourglass in- and out-cones assures that each edge of $Q$ makes
an angle within $\pm 45^\circ$ of the horizontal, so the halfplanes comprising
$R(Q)$ make angles within $\pm 45^\circ$ of the vertical. Thus their
intersection includes the out-cone at $u_k$.
\end{proof}

\noindent
The import of this lemma is that an hourglass path can be extended inward
with any edge in the in-cone of the hourglass at $u_k$, and remain radially monotone.

So far we have ignored the triangulation mesh of the halfplane $C$.
A triangulation is called a \emph{non-obtuse triangulation} if no
triangle angle strictly exceeds $\pi/2$, i.e., is obtuse.
If the halfplane $C$ is non-obtusely triangulated, then the in-cone at $u_k$
necessarily contains some triangulation edge (because the cone
angle is $90^\circ$).

Now imagine executing Algorithm~1 on this halfplane $C$, 
except with the concentric circles replaced by the vertical lines through each
vertex of the triangulation (effectively, infinite-radii concentric circles).
Suppose we have grown a forest that connects every vertex strictly right of $v$ to
$\bC$, and we seek to connect $v$ to grow this forest.
It should be clear that the out-cone of the hourglass at $v$ must contain an edge
of the non-obtuse triangulation, and so can indeed be connected by
extending a radially monotone path.

\subsection{Hourglass Paths for Convex $C$}
\seclab{HourGlassC}
Now we turn back to the less contrived situation of a convex domain $C$,
but again non-obtusely triangulated.
Every planar straight-line graph on $n$ vertices has a conforming non-obtuse triangulation
of $O(n^{2.5})$ triangles~\cite{b-nonobtri-16}.
We will follow Algorithm~1 just as in the halfplane case, but now
the proofs are no longer straightforward.
A sample result of applying the algorithm 
is shown in Fig.~\figref{NonObcocircs}.
\begin{figure}[htbp]
\centering
\includegraphics[width=0.65\linewidth]{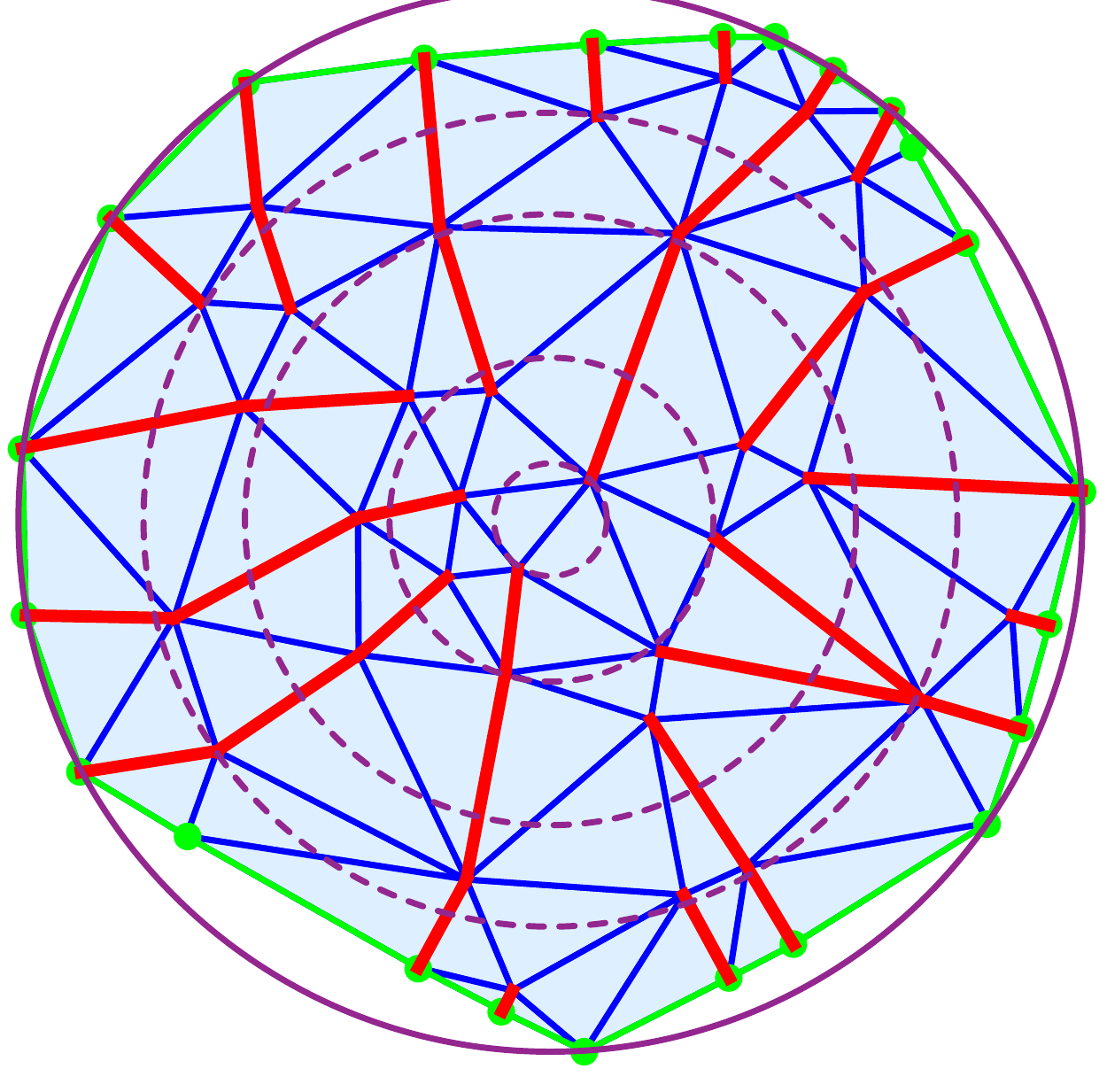}
\caption{A $56$-vertex example, with several concentric circles shown.
The cut forest $\F$ is marked (red).}
\figlab{NonObcocircs}
\end{figure}

\subsection{Theorem Statement}
Although I have no counterexample for arbitrary convex domains $C$,
the proof below restricts $C$ to those that ``nicely" fit within
the minimal bounding circle $B$.
Define a convex domain $C$ as \emph{round} if
the neighborhood of every vertex $v \in \bC$
contains the in-cone of the hourglass $H$ at $v$,
where the baseline of $H$ is tangent to the circle through $v$
centered at the center $x$ of $B$.
\begin{theorem}
Let $C$ be a non-obtusely triangulated
round convex domain.
Then $C$ has an radially monotone spanning forest.
\thmlab{NonObConcirc}
\end{theorem}

\subsection{Proof of Theorem~\protect\thmref{NonObConcirc}}
Let $C$ be a round convex domain, with bounding circle $B$
centered on $x$.
We define the hourglass $H$ at a vertex $v$ to have
baseline tangent to the circle through $v$ centered on $x$.
Let $Q=(u_0,\ldots,u_k)$ be a path of edges in $C$, with $u_0 \in \bC$.
We say that $Q$ is an \emph{hourglass path} if
each edge $u_i u_{i+1}$ falls within the in-cone at $u_i$ and
the out-cone at $u_{i+1}$.
See Fig.~\figref{Hgpath}.
\begin{figure}[htbp]
\centering
\includegraphics[width=0.90\linewidth]{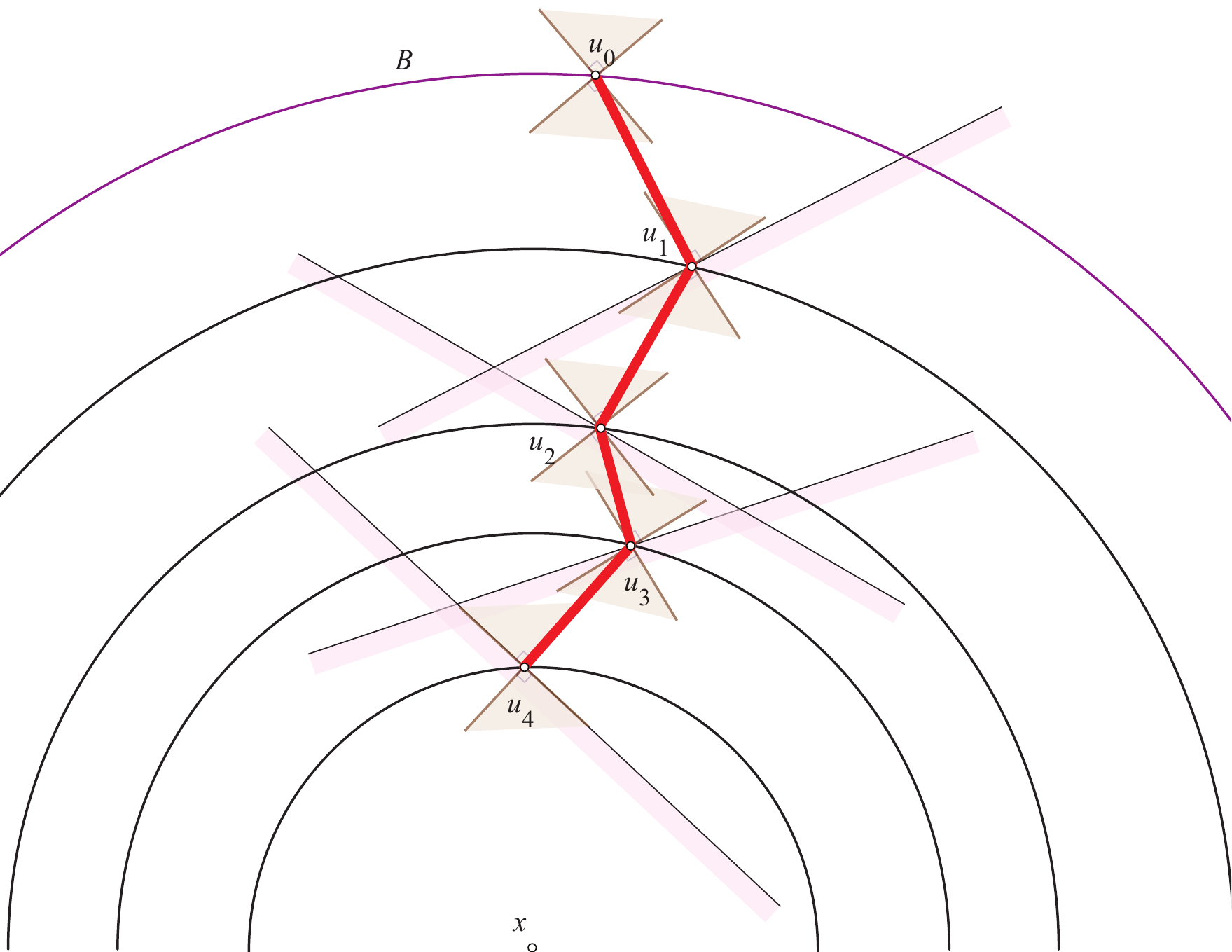}
\caption{An hourglass path walks between concentric circles.}
\figlab{Hgpath}
\end{figure}
We need to prove three results to achieve Theorem~\thmref{NonObConcirc}:
\begin{enumerate}
\squeezelist
\item An hourglass path is radially monotone: Lemma~\lemref{hgrm}.
\item The in-cone at the last vertex $u_k$ is inside the region $R(Q)$,
so that $Q$ can be extended inward: Lemma~\lemref{NoClipping}.
\item For the next vertex $v$ to be processed, any connection
within its out-cone to an earlier processed vertex $v'$ falls within
the in-cone of $v'$: Lemma~\lemref{OutconeIncone}.
\end{enumerate}

\begin{lemma}
Let $v$ be a vertex at distance $1$ from $B$'s center $x$.
Then for every vertex $v'$ within the out-cone of $v$
at distance $r > 1$ from $x$,
the in-cone of $v'$ includes $v$.
\lemlab{OutconeIncone}
\end{lemma}
\begin{proof}
Refer to Fig.~\figref{HgOutCone}. The vertices within the out-cone of $v$
each cover $v$ with their in-cones.
The extreme case occurs when $r \to 1$ and a vertex $v'$ lies on the boundary
of the out-cone of $v$. For all $v'$ within the out-cone of $v$, the edge $v v'$ lies in
both the out-cone of $v$ (by assumption) and within the in-cone of $v'$.
\end{proof}
\begin{figure}[htbp]
\centering
\includegraphics[width=0.75\linewidth]{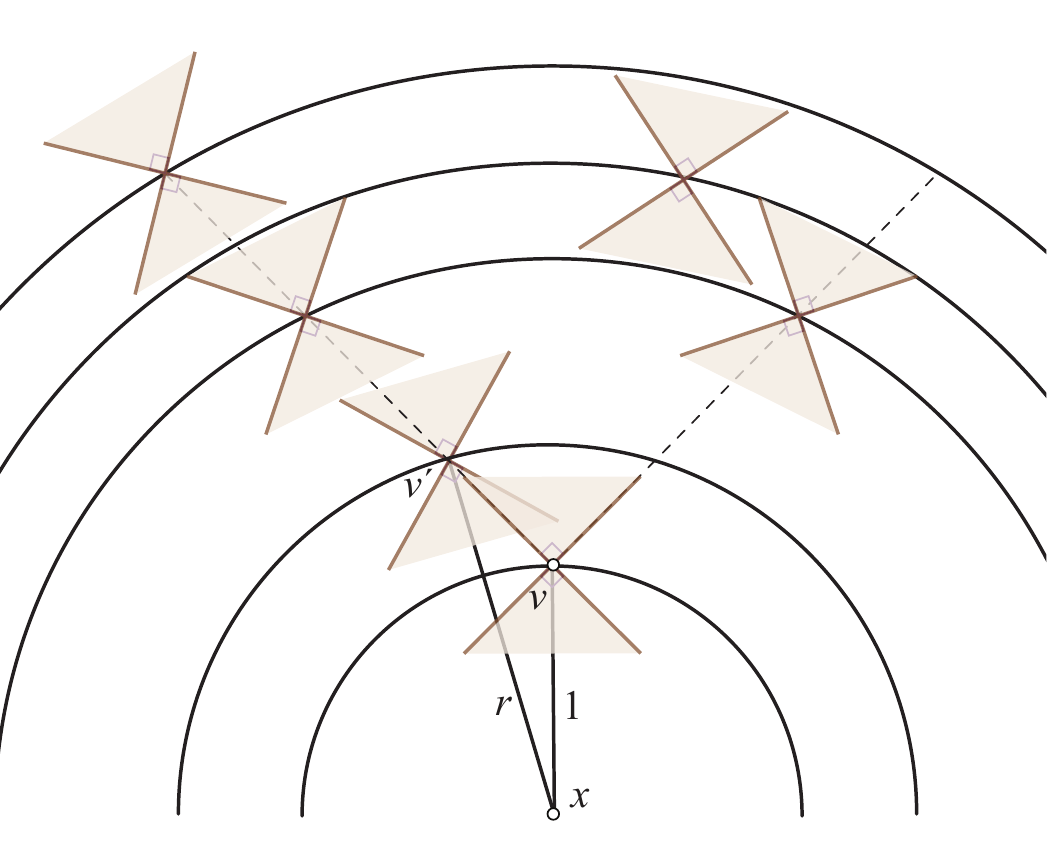}
\caption{$v$ is covered by the in-cone of any vertex $v'$ 
that falls within the out-cone of $v$.}
\figlab{HgOutCone}
\end{figure}
The consequence of this lemma is that a connection to an hourglass path
that terminates at $v'$ by the edge $v v'$ will maintain the grown path as
an hourglass path.

The next two lemmas seem technically difficult, and my
current arguments depend on numerical computations,
but numerical arguments which I believe are convincing
and could be converted to analytical proofs with effort.
I will continue to call them ``lemmas" with the understanding they are 
(justified) claims at this point.

\begin{lemma}
Hourglass paths are radially monotone.
\lemlab{hgrm}
\end{lemma}
\begin{proof}
Let $Q$ be an hourglass path, ending at innermost vertex $a$, which lies
on a circle of radius $1$. Let $b$ be an arbitrary point on $Q$.
We need to prove that the angle $ab$ makes with the next edge $bc$ of $Q$
beyond (outward) of $b$ is at most $90^\circ$. We know that edge
lies in the out-cone of $b$, because $Q$ is an hourglass path.
The situation is illustrated in Fig.~\figref{HgRmAngCalc}.
From that figure, we need to show that $\b \le 45^\circ$
for any $r = |ab|$.

Fig.~\figref{HgRmPlot1} plots the angle $\b$ as a function of $r$ and $\a$,
which two parameters (together with the assumption w.l.o.g. that
$|ax|=1$) completely determine $\b$.
It is clear that in fact $\b$ can exceed $45^\circ$.
However, the region of $(r,\a)$ values in which $\b > 45^\circ$
cannot occur with an hourglass curve.

The extreme values of $\a$ occur when $Q$ is a $45^\circ$-spiral,
that is, when the path turns as much as possible clockwise within the
confines of the hourglass constraints at each vertex. The maximum
turn occurs in the smooth case, when the spacing between the vertices along $Q$
approaches zero.
It is easy to compute the $(r,\a)$ values achieved by this $45^\circ$-spiral.
Fig.~\figref{HgRmPlot2} shows that only at $r=1$ and $\a=45^\circ$
does the spiral values touch the boundary of the forbidden $\b$ region.
Other less extreme $Q$ avoid the forbidden region entirely.
The figure also shows that other hourglass paths are less extreme
than the $45^\circ$-spiral.
\end{proof}
\begin{figure}[htbp]
\centering
\includegraphics[width=0.75\linewidth]{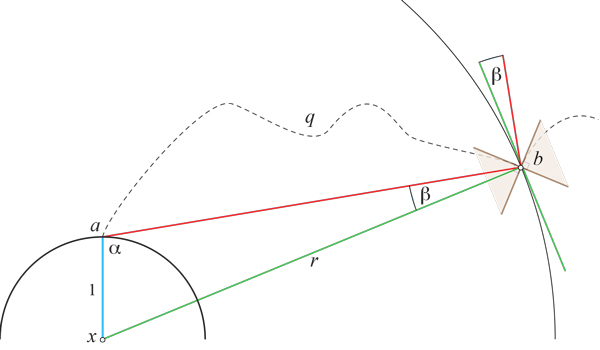}
\caption{Lemma~\protect\lemref{hgrm} is proved if $\b \le 45^\circ$.}
\figlab{HgRmAngCalc}
\end{figure}

\begin{figure}[htbp]
\centering
\includegraphics[width=0.75\linewidth]{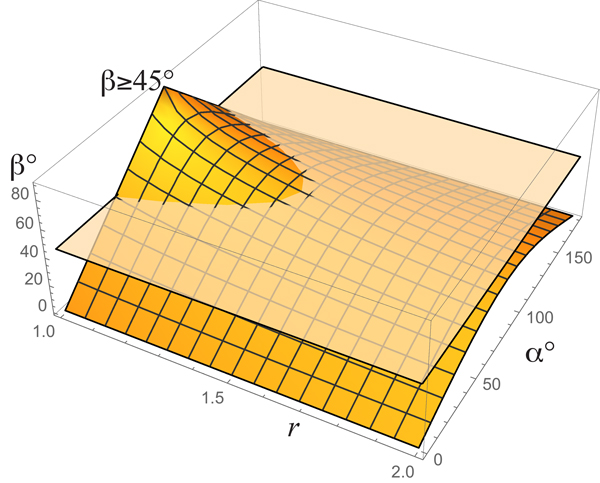}
\caption{The $\b \ge 45^\circ$ region for all $(r,\a)$ combinations.}
\figlab{HgRmPlot1}
\end{figure}

\begin{figure}[htbp]
\centering
\includegraphics[width=0.75\linewidth]{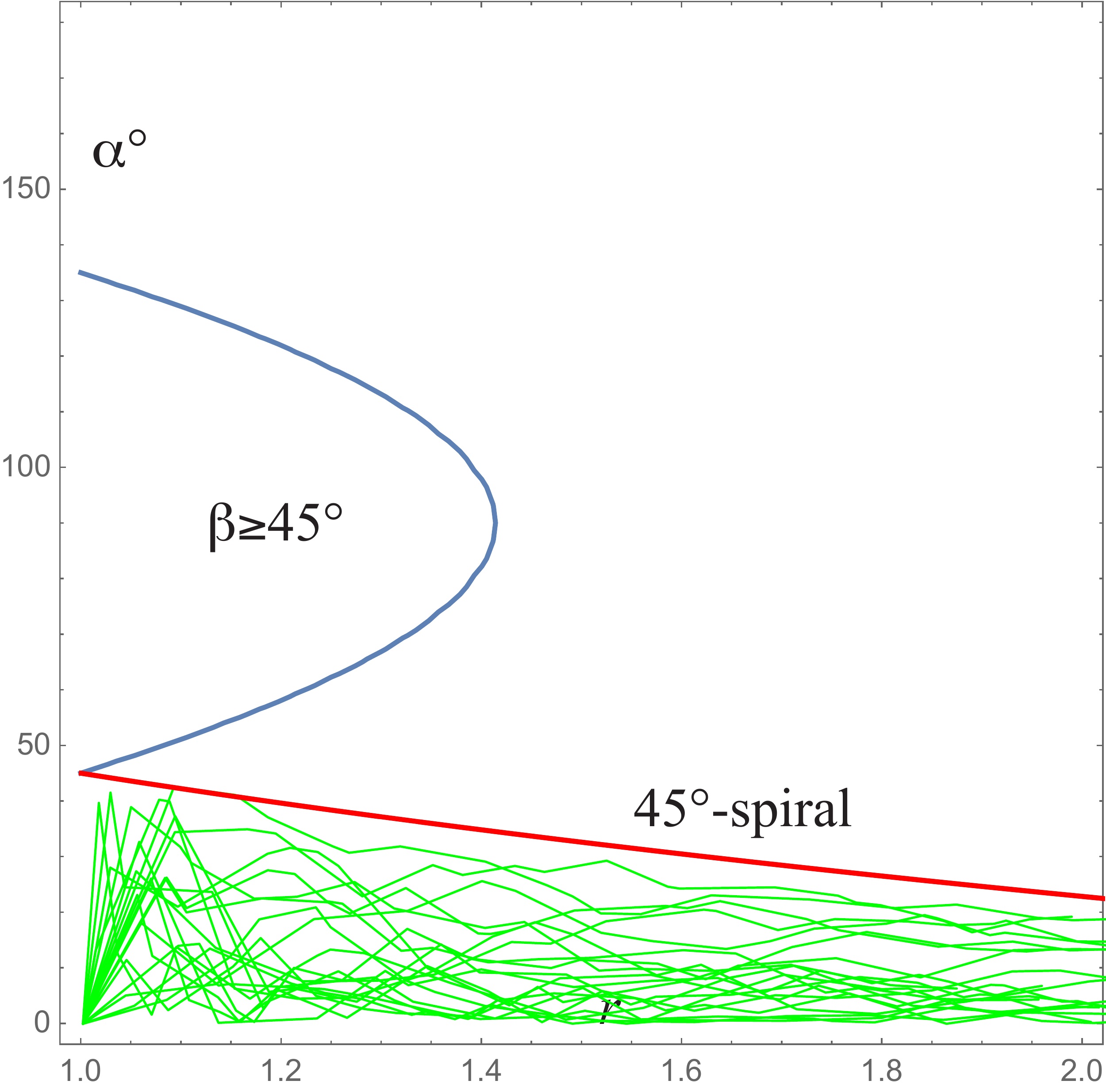}
\caption{Hourglass paths never enter the $\b \ge 45^\circ$ region.
The green curves show $(r,\a)$ pairs for random hourglass paths.}
\figlab{HgRmPlot2}
\end{figure}

Finally we show that the in-cone of the last vertex of an hourglass curve
is not ``clipped," which means it can be extended.
\begin{lemma}
Let $Q$ be an hourglass curve, ending at innermost vertex $a$.
Then the region $R(Q)$ includes the in-cone of the hourglass at $a$
as far as the center $x$.
\lemlab{NoClipping}
\end{lemma}
\begin{proof}
We assume the same setup as in the previous lemma:
$|ax|=1$.
Again the issue is decided by the extreme curve, a $45^\circ$-spiral.
Fig.~\figref{ConCircCalc} illustrates the challenge: each edge along the spiral
generates a halfplane constraint that contributes to $R(Q)$.
The triangle subset of the in-cone of $a$ must not be clipped by any of
these halfplanes--i.e., the halfplanes must include corner $c$ of the triangle--or future inward growth of the hourglass path could
be compromised.

For a point $b$ at polar coordinates $(r,\q)$, $r=s e^{\pi/2-\q}$, we need the angle 
$\b = \angle xbc \le 45^\circ$.
(Note: this $\b$ is not the same as $\b$ in the previous lemma.)
The angle $\b$ ``consumes" a portion of the $45^\circ$ between $r$ and the halfplane.
Thus only if $\b > 45^\circ$ does the halfplane clip $c$.
Fig.~\figref{HgRmAngPlot} shows that $\cos \b \ge \sqrt{2}/2$, only
equaling that value when $\q=90^\circ$, when $a=b$.
This makes sense, as the first tangent to the spiral curve at point $a$
is at $45^\circ$ w.r.t. the horizontal, and then the halfplane passes through,
but does not clip off, corner $c$.
\end{proof}

\noindent
That the same degeneracy $a=b$ occurs in both
of the previous two lemmas suggests there may be a uniting viewpoint.

\begin{figure}[htbp]
\centering
\includegraphics[width=1.0\linewidth]{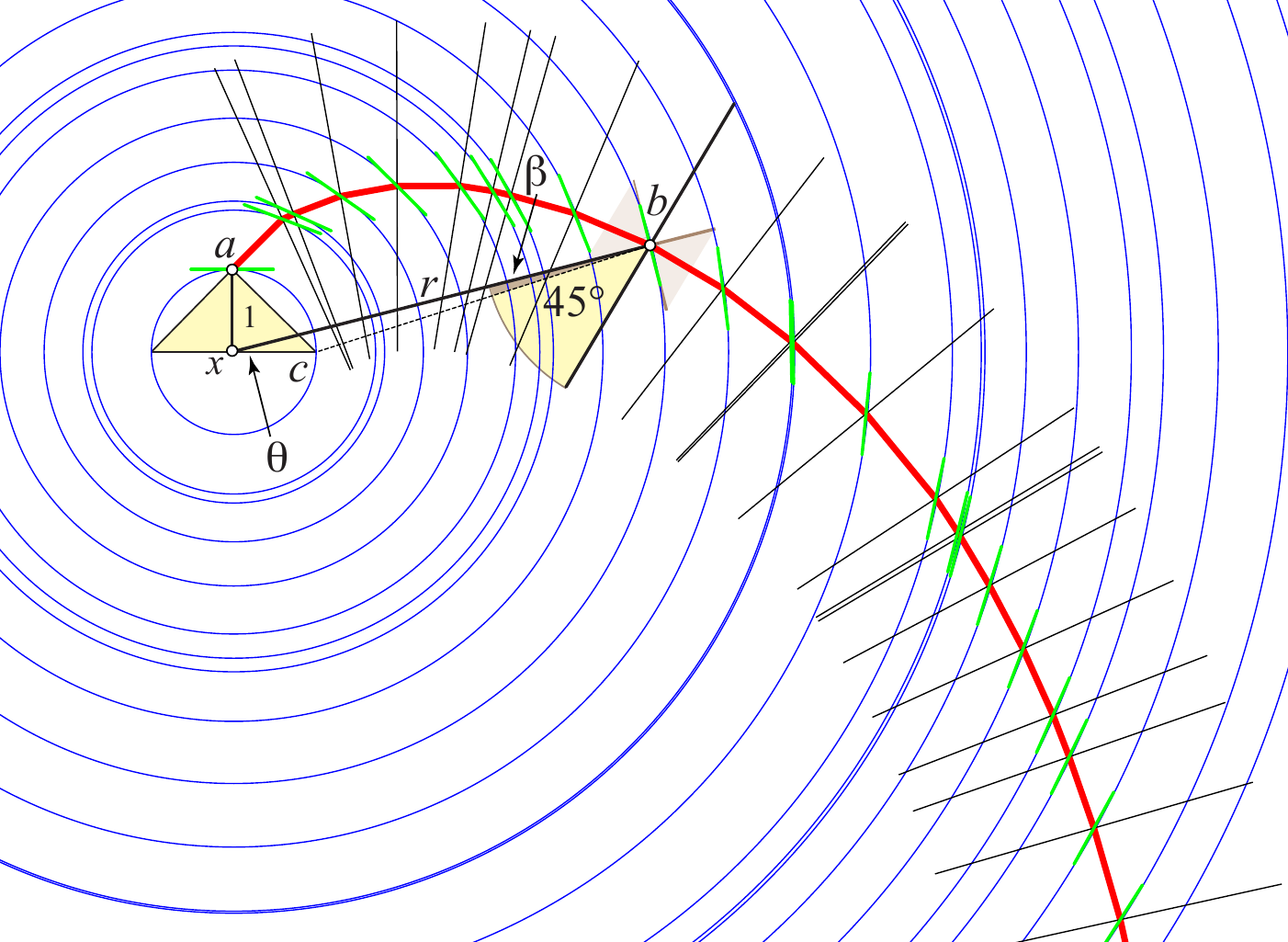}
\caption{Circle tangents are green. Orthogonal halfplane boundaries black.
The halfplanes do not clip corner $c$.}
\figlab{ConCircCalc}
\end{figure}


\begin{figure}[htbp]
\centering
\includegraphics[width=0.75\linewidth]{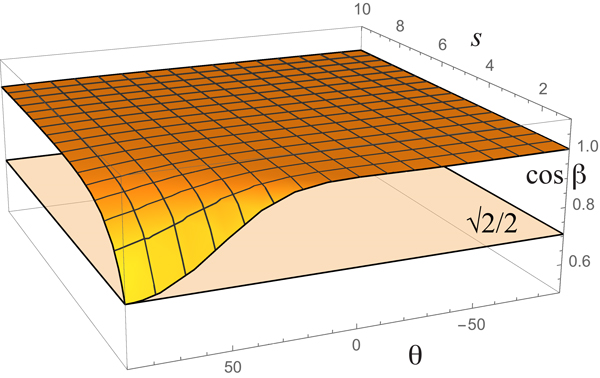}
\caption{When $\cos \b > \sqrt{2}/2$, $\b < 45^\circ$.}
\figlab{HgRmAngPlot}
\end{figure}

Now we have proven Theorem~\thmref{NonObConcirc} which
proves that Algorithm~1 works on round, non-obtusely triangulated
convex domains: the algorithm always finds a spanning forest
composed of hourglass paths, which are radially monotone.
The restriction to round convex domains was imposed to enable
the first edge of a path starting from a vertex on $\bC$ to fall within the hourglass in-cone
there. It is likely that this roundness assumption is not necessary.
Moreover, I believe even the convex assumption is not necessary:
likely only star-shapedness from $x$ is needed.

The insistence on a non-obtuse triangulation, however, is crucial,
as the example in Appendix~1 demonstrates.

One might wonder whether the $45^\circ$-spiral, which plays
a prominent role in the analysis, can actually occur
in a non-obtuse triangulation. The answer is {\sc{yes}}:
see Fig.~\figref{OctSpiral10}.
\begin{figure}[htbp]
\centering
\includegraphics[width=0.75\linewidth]{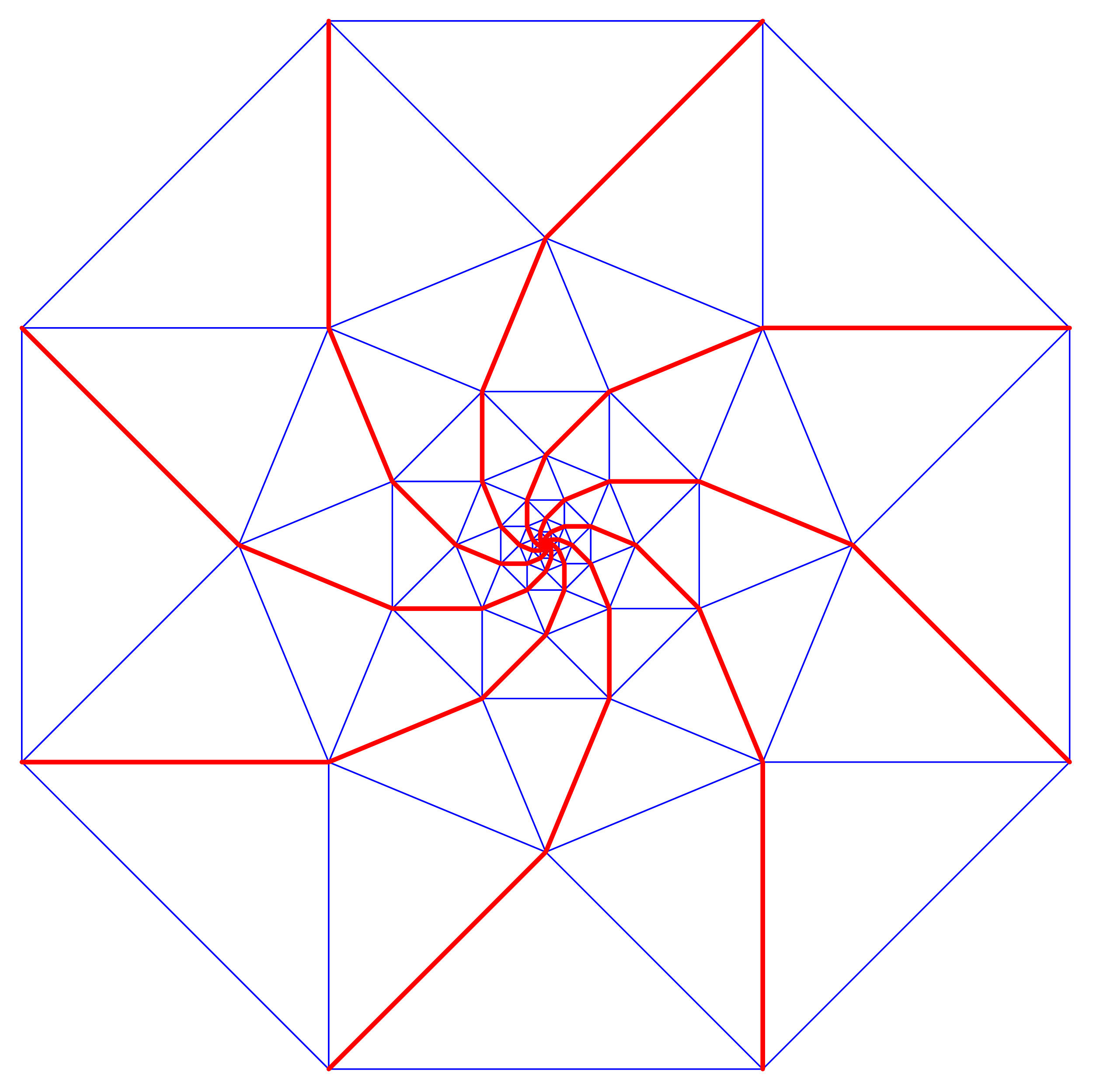}
\caption{Radially monotone $45^\circ$-spiral spanning forest in
a non-obtuse triangulation.}
\figlab{OctSpiral10}
\end{figure}

\subsection{More Examples}
Several more examples are shown in Fig.~\figref{NonOb4Ex}.
\begin{figure}[htbp]
\centering
\includegraphics[width=1.0\linewidth]{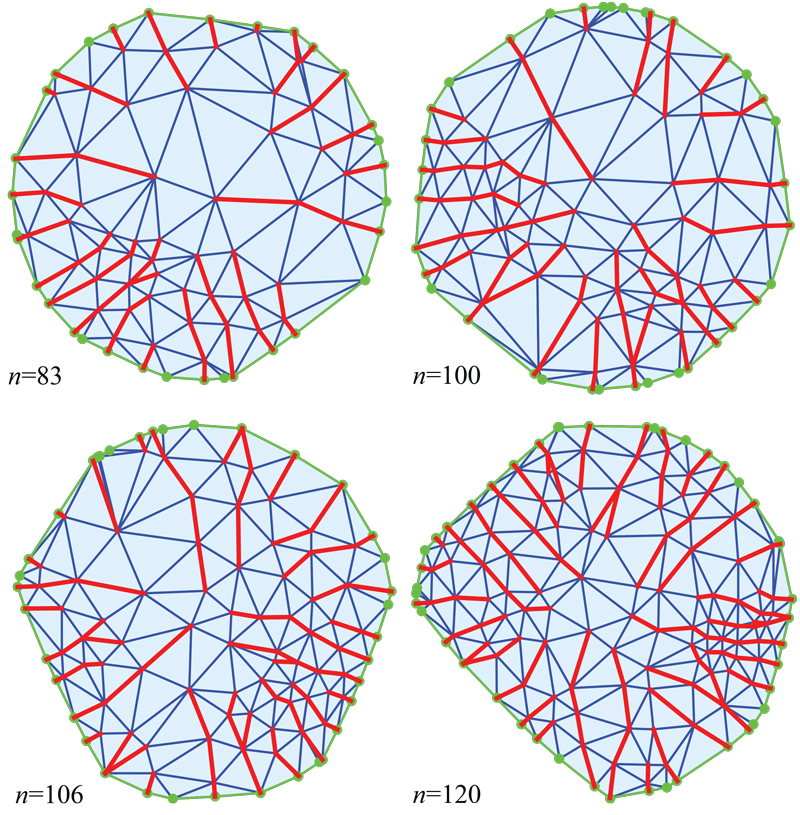}
\caption{Radially monotone spanning forests found by Algorithm~1.}
\figlab{NonOb4Ex}
\end{figure}

\section{Radial Monotone Paths in 3D}
\seclab{rm3D}
Although the whole point of radial monotone paths is to cut open
polyhedra via such paths, I have yet to define what the notion
means on a polyhedron in $\mathbb{R}^3$.
There is more than one way to generalize the notion,
but I have focused on one generalization that I find natural.
We assume henceforth that $C$ is a convex cap with boundary $\bC$ a topological
circle, i.e., $C$ is a simply connected subset of the faces of a convex polyhedron.
The earlier Fig.~\figref{CutsLay_s2_pi3_n500} shows such a convex cap.
All of our empirical explorations also assume that the polyhedron from which
$C$ is derived is spherical: all vertices on a sphere $\S$.

Under these circumstances, we define a path $Q$ of edges of $C$,
$Q=(v_0,\ldots,v_k)$ with $v_k \in \bC$ to be radially monotone
if the \emph{medial path} $M(Q)$ is radially monotone.
The next section explains and justifies this definition.

\subsection{$LMR$-Chains}
\seclab{LMR-Chains}
Let $\o_i$ be the curvature at $v_i$.
As usual, we view $v_0$ as a leaf of a cut forest, which will then serve as
the end of a cut path, and the ``source" of opening that path.

Let $\l_i$ be the angle at $v_i$ left of $Q$, and $\r_i$ the angle right of $Q$ there.
So $\l_i + \o_i + \r_i = 2 \pi$.
Define $L$ to be the planar path from the origin with left angles $\l_i$,
$R$ the path with right angles $\r_i$, and
$M$ the \emph{medial path} with left angles $\l_i - \o_i/2$
(and therefore right angles  $\r_i + \o_i/2$).
(Each of these paths are understood to depend on $Q$: $L=L(Q)$ etc.)
We label the vertices of the paths $\ell_i, m_i, r_i$, 
with $m_0 m_1$ on the $x$-axis.
See Fig.~\figref{LMREx1}.
\begin{figure}[htbp]
\centering
\includegraphics[width=0.75\linewidth]{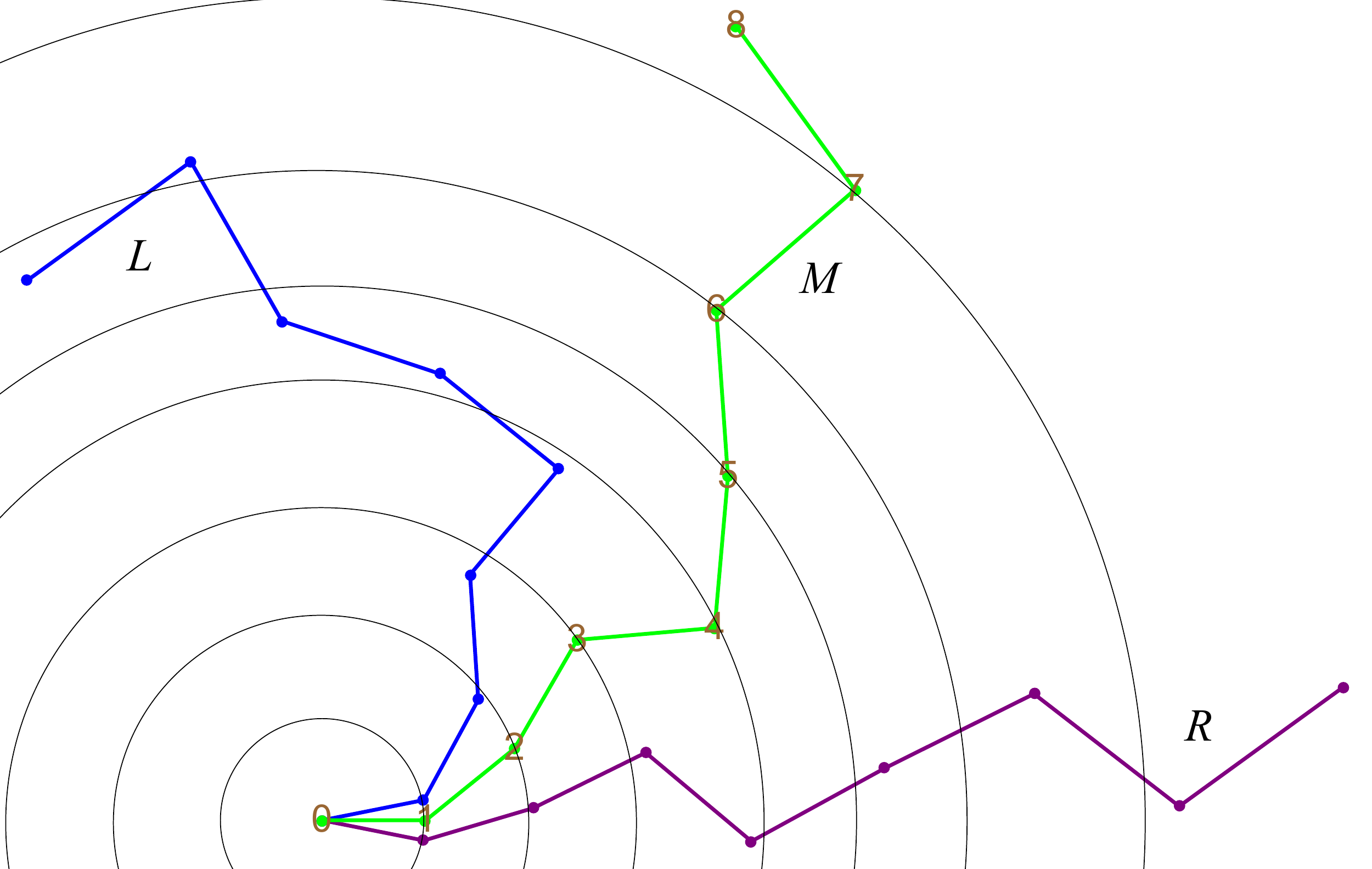}
\caption{The three paths $L,M,R$. Here $M$ is radially monotone but $L$ is not.}
\figlab{LMREx1}
\end{figure}

The intuition behind the focus on $M(Q)$ is two-fold:
(1)~It is important to allow either or both of $L$ and $R$ to be non-rm;
(2)~One can prove that $L \cap R = \{v_0\}$ for much more than infinitesimal
openings caused by the curvature encountered along $Q$.

Let $\O=\sum_i \o_i$ along $Q$: The total curvature of all the vertices on the cut path.
Let $\t(M)$ be the maximum $\pm$ turn of the edge vectors in $M$: 
the largest absolute value the vector
$m_i-m_{i-1}$ makes with the $x$-axis.

Our goal is to prove this theorem:
\begin{theorem}
If path $M$ is radially monotone, and in addition,
$\t(M) \le \pi/2$ and $\O \le \pi$, then 
none of the three paths $L$, $M$, and $R$ cross one another.
\thmlab{LMR-nonint}
\end{theorem}

\noindent
The reason for the angle restrictions is that, without some restrictions, the claim could be false.
Fig.~\figref{LMRSpiralCross}
shows a $70^\circ$-spiral $M$ with $\t(M) > (3/2)\pi$ and
$\O=2 \pi$, where $L$ crosses $R$.
\begin{figure}[htbp]
\centering
\includegraphics[width=0.5\linewidth]{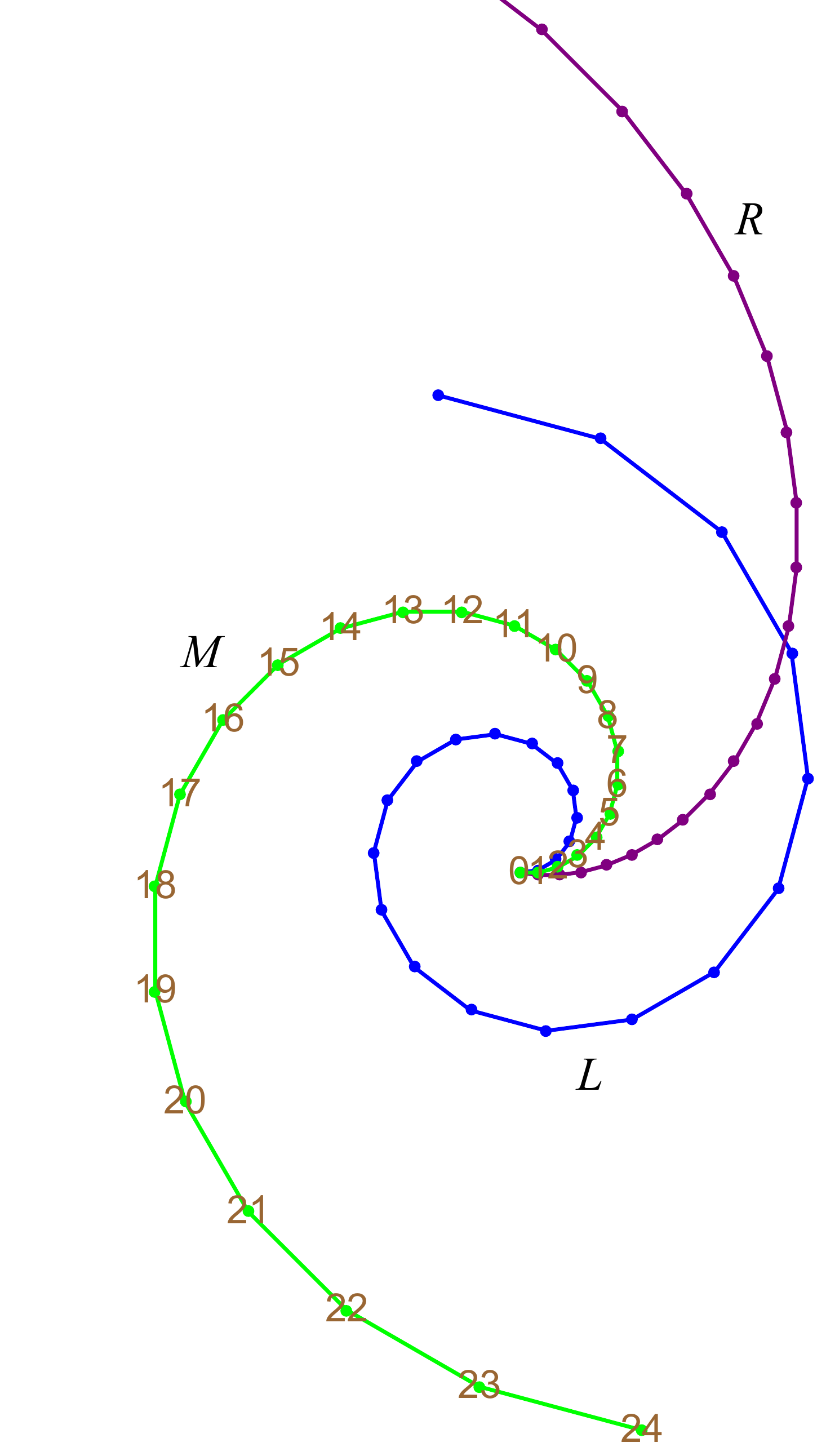}
\caption{A $70^\circ$-spiral $M$, with $\O=2 \pi$.
$L$ crosses $R$ right-to-left.}
\figlab{LMRSpiralCross}
\end{figure}
Restricting $\t(M)< \pi$ and $\O \le \pi$ avoids crossing;
see Fig.~\figref{LMRSpiralPi}.
\begin{figure}[htbp]
\centering
\includegraphics[width=0.75\linewidth]{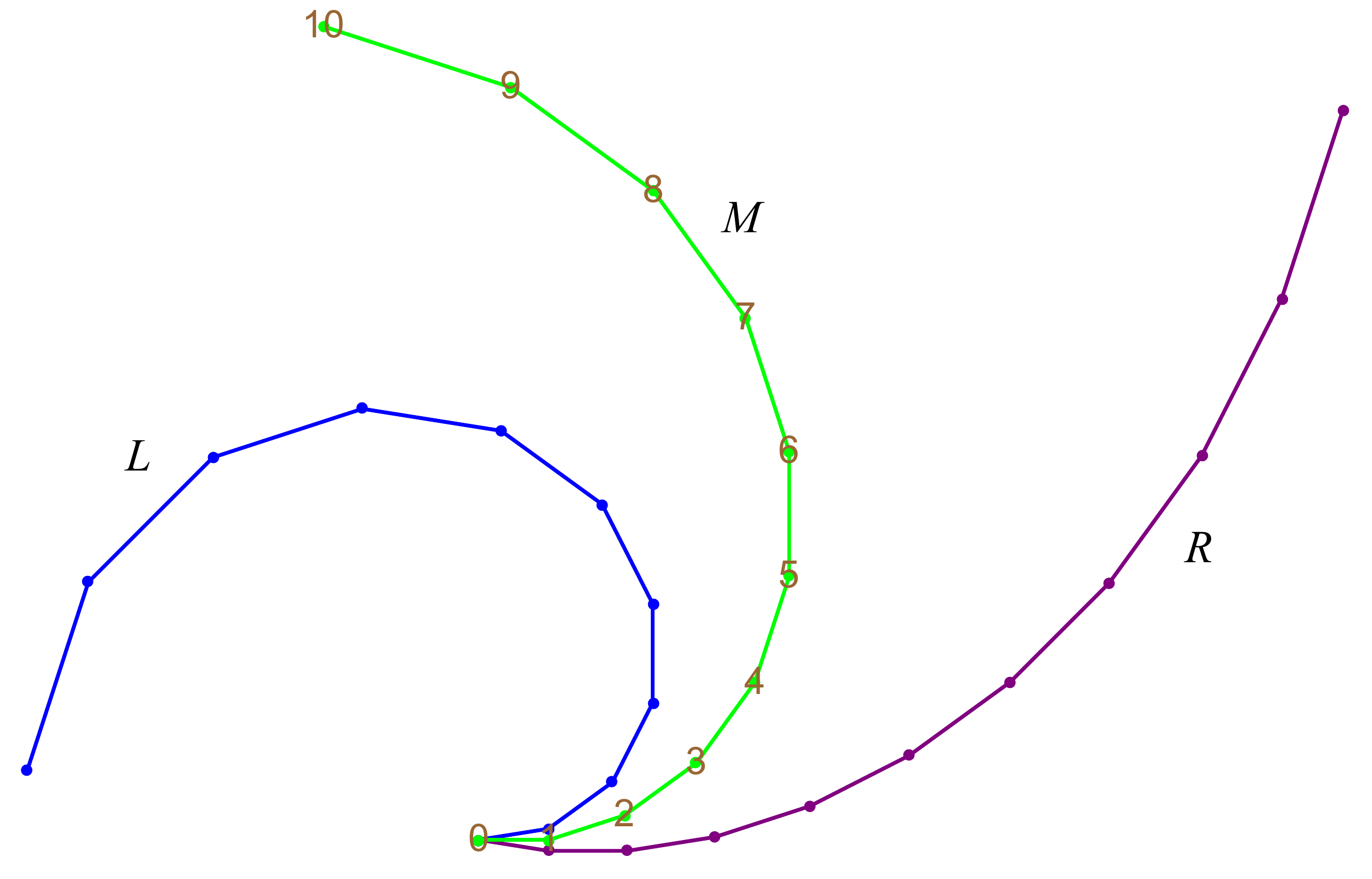}
\caption{ $\O = \pi$ so the last tangents to $L$ and $R$ are $180^\circ$ apart.
Here $\t(M)< \pi$.}
\figlab{LMRSpiralPi}
\end{figure}
I do not believe the angle restrictions are close to tight.
For example, the theorem seems to hold for 
$\t(M)< 2\pi$ and $\O \le \pi$, and for $\t(M)< \pi$ and $\O \le 2\pi$.
The severe angle restrictions are imposed to obtain a clean proof.
And they seem to suffice for non-overlap of random spherical polyhedra.

We prove Theorem~\thmref{LMR-nonint} via two lemmas.
\begin{lemma}
Under the conditions stated,
$L$ cannot cross $M$ left-to-right.
\lemlab{LM-noint}
\end{lemma}
\begin{proof}
(Crossing right-to-left will be excluded by the next lemma.)
The proof is by induction.
It clearly holds for $k=1$, when $M$ is just a single edge,
and $L$ is rotated $\o_0$ ccw of $M$.

Suppose the claim is true for $k$.
So the situation is like that shown in Fig.~\figref{LMREx1}.
Imagine extending $M$ prior to $v_0$, to a vertex $v'$.
Then the $L$ path is rigidly rotated ccw by $\o'$.
The new path $L'$ is strictly further away from $M$ along the concentric
circles centered on $v'$.
See Fig.~\figref{LMREx1a}.
\end{proof}
\begin{figure}[htbp]
\centering
\includegraphics[width=0.75\linewidth]{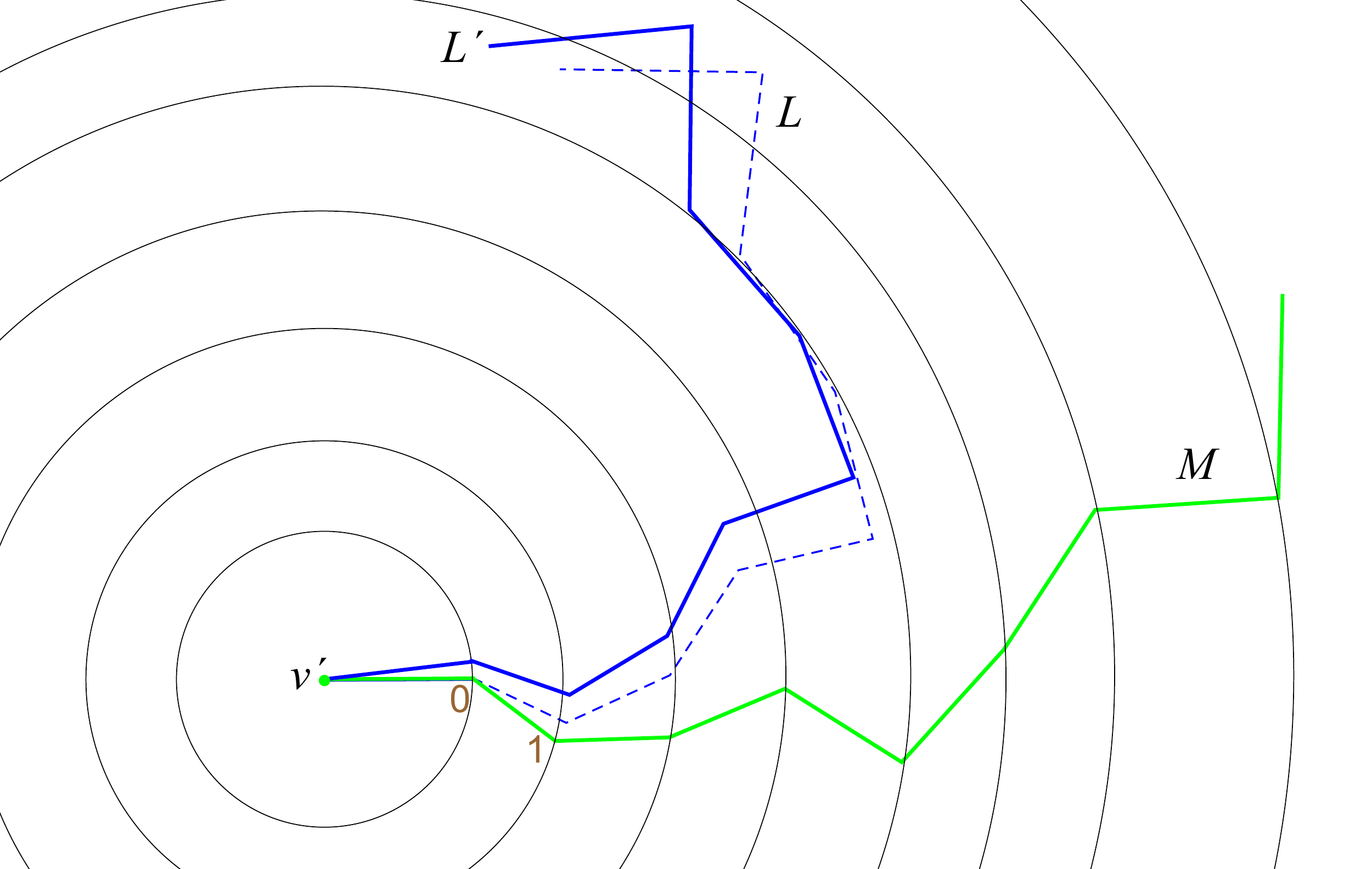}
\caption{Extending the $L$ and $M$ paths in Fig.~\protect\figref{LMREx1}
by one edge $v' v_0$. ($R$ not shown.)}
\figlab{LMREx1a}
\end{figure}

\begin{lemma}
Under the conditions stated,
$L$ cannot cross $R$ right to left.
\lemlab{LR-noint}
\end{lemma}
\begin{proof}
With the convention that $m_1- m_0$ aims along the $+x$-axis, 
the restriction
$\t(M) \le \pi/2$ implies that all the edge vectors of $M$ point
in the $+x$-halfplane.
The restriction $\O \le \pi$ turns the vectors of $M$  to point
into the $+y$-halfplane for $L$; 
recall $\l_i$ is turned ccw $\o_i/2$ from the
corresponding $M$ angle.
Similarly the vectors of $M$ are turned cw  $\o_i/2$  and so point
into the $-y$-halfplane.
Thus $L$ lies entirely in the upper halfplane and $R$ in the lower, so they
cannot cross.
\end{proof}

\noindent
These two lemmas prove Theorem~\thmref{LMR-nonint}.

\section{Algorithm~2: Finding a Radially Monotone Cut Tree}
\seclab{Algorithm2}
Using the definition in the previous section, we now suggest an
``algorithm" to find a radially monotone cut tree for a convex polyhedron.
The algorithm does not provably succeed, even on spherical non-obtusely
triangulated polyhedra. But its strong empirical performance suggests some
variant might provably succeed.
For now, I just describe the decisions made to achieve a definite algorithm.

There are five decisions that drive the algorithm:
\begin{enumerate}
\squeezelist
\item Attention is restricted to spherical polyhedra $\P$, vertices inscribed
in a sphere $\S$.
\item Latitude circles on $\S$ are used to mimic the concentric circles
employed in the 2D Algorithm~1.
\item The algorithm constructs trees from paths $Q$ on $\P$ whose planar medial paths $M(Q)$ are radiallly monotone.
\item The full $\P$ is treated as a convex cap $C$ with boundary $\bC$.
\item Hourglass paths are abandoned.
\end{enumerate}
We now remark on each of these decisions.

\paragraph{(1)}Spherical polyhedra are special because of their relationship to
Delaunay triangulations, a point to which we will return later.
Outside of this class, the algorithm is far less successful.
For example, it fails to find a rm cut tree for
more than $50$\% of $100$-vertex random polyhedra
inscribed in an ellipsoid with axes $(1,1,\frac{1}{4})$.
See Fig.~\figref{Cuts_Lay_s19_n100_z25_Overlap}.
\begin{figure}[htbp]
\centering
\includegraphics[width=1.0\linewidth]{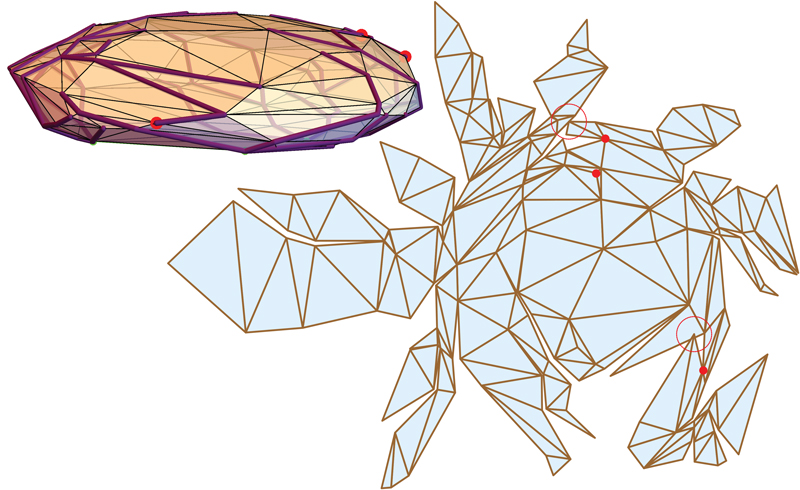}
\caption{Overlap. The $100$-vertex polyhedron is drawn from an ellipsoid with
axes $(1,1,\frac{1}{4})$. The three marked vertices could not be 
joined radially monotonically to
the cut forest. Two of the non-rm cuts cause overlap in the unfolding.}
\figlab{Cuts_Lay_s19_n100_z25_Overlap}
\end{figure}

\paragraph{(2)}In order to mimic the successful 2D Algorithm~1, we process the vertices
in order of their geodesic distance from the north pole $N$,
furthest first;
see Fig.~\figref{SphereLatitude}.
We view $v$ as lying on a planar circle centered on $N$ with radius $\g$.
\begin{figure}[htbp]
\centering
\includegraphics[width=0.5\linewidth]{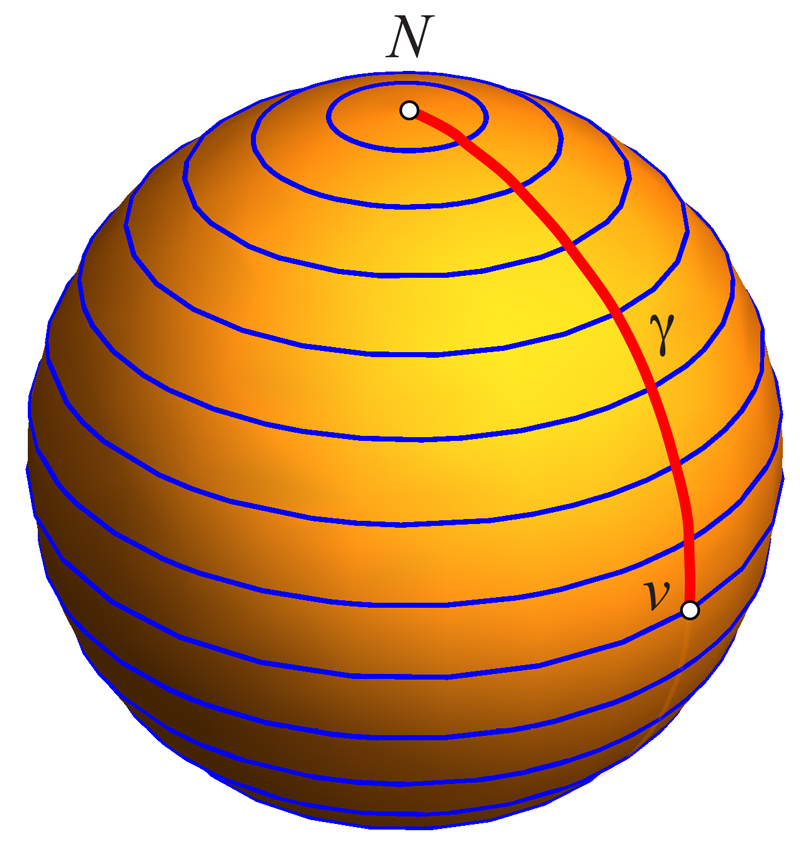}
\caption{Concentric circles representing latitude lines on $\S$.}
\figlab{SphereLatitude}
\end{figure}

\paragraph{(3)}We have already justified the concentration on medial rm cut paths
in Section~\secref{LMR-Chains}:
the unfolded two sides of the cut do not intersect
(Theorem~\thmref{LMR-nonint}).

\paragraph{(4)}Fig.~\figref{CutsLay_s2_pi3_n500} showed a convex cap
with a jagged boundary $\bC$. One can extend a convex cap of polyhedron $\P$
by removing just one triangle $\triangle abc$ from $\P$:
then the convex ``cap" is all but that one triangle.
Even further, one can remove an infinitesimal slice around the
boundary of that triangle, the nonconvex degenerate quadrilateral
$\bC= (a,b,c,b)$. Then $C$ includes every face of the full $\P$.

Returning to our first figure, Fig.~\figref{CutsLay_s1_n200}, the green V-shape
represents that infinitely thin quadrilateral $\bC$.
The rm cut forest $\F$ is grown from $\bC$, and connected to a tree
by the two edges $ab$ and $bc$. 

For our experiments with random spherical polyhedra, we added north-
and south-pole points $N$ and $S$,
just for coding convenience. $N$ plays the role of the 2D bounding
circle center $x$ in Algorithm~1,
and $S$ is a vertex of $\triangle abc$.
This is illustrated more clearly in
Fig.~\figref{CutsLay_s1_n50},
which shows the three trees grown from $\bC$.
\begin{figure}[htbp]
\centering
\includegraphics[width=1.0\linewidth]{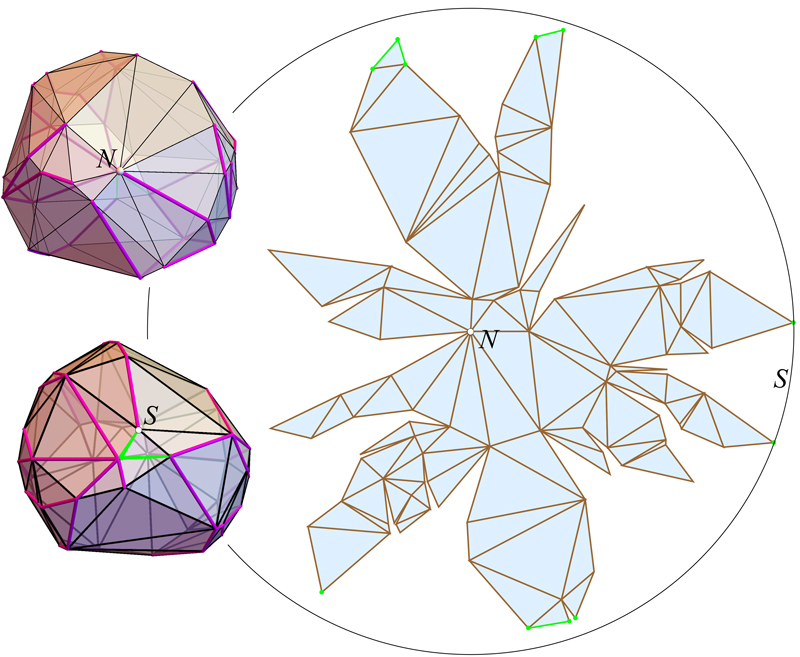}
\caption{$\P$ has $n=50$ vertices. $N$ and $S$ label the north and south poles respectively.}
\figlab{CutsLay_s1_n50}
\end{figure}

\paragraph{(5)}The reason I abandoned hourglass paths is that, under
various generalizations, they did not always seem to exist.
They were a specific tool to prove that Algorithm~1 finds rm forests in 2D,
but seem less useful in 3D.

\subsection{More Examples}
Figs.~\figref{CutsLay_s2_n300}, \figref{Lay_s2_n500},
\figref{Lay_s2_n1000}, and~Figs.~\figref{Cuts_s2_n1500} and \figref{Lay_s2_n1500},
show four more example unfoldings.
\begin{figure}[htbp]
\centering
\includegraphics[width=1.0\linewidth]{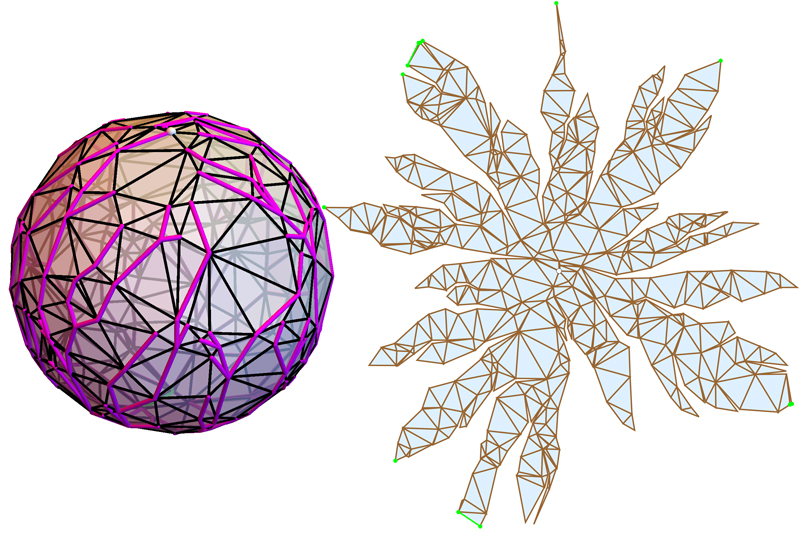}
\caption{Cut tree and unfolding of a polyhedron of $n=300$ vertices.}
\figlab{CutsLay_s2_n300}
\end{figure}
\begin{figure}[htbp]
\centering
\includegraphics[width=1.0\linewidth]{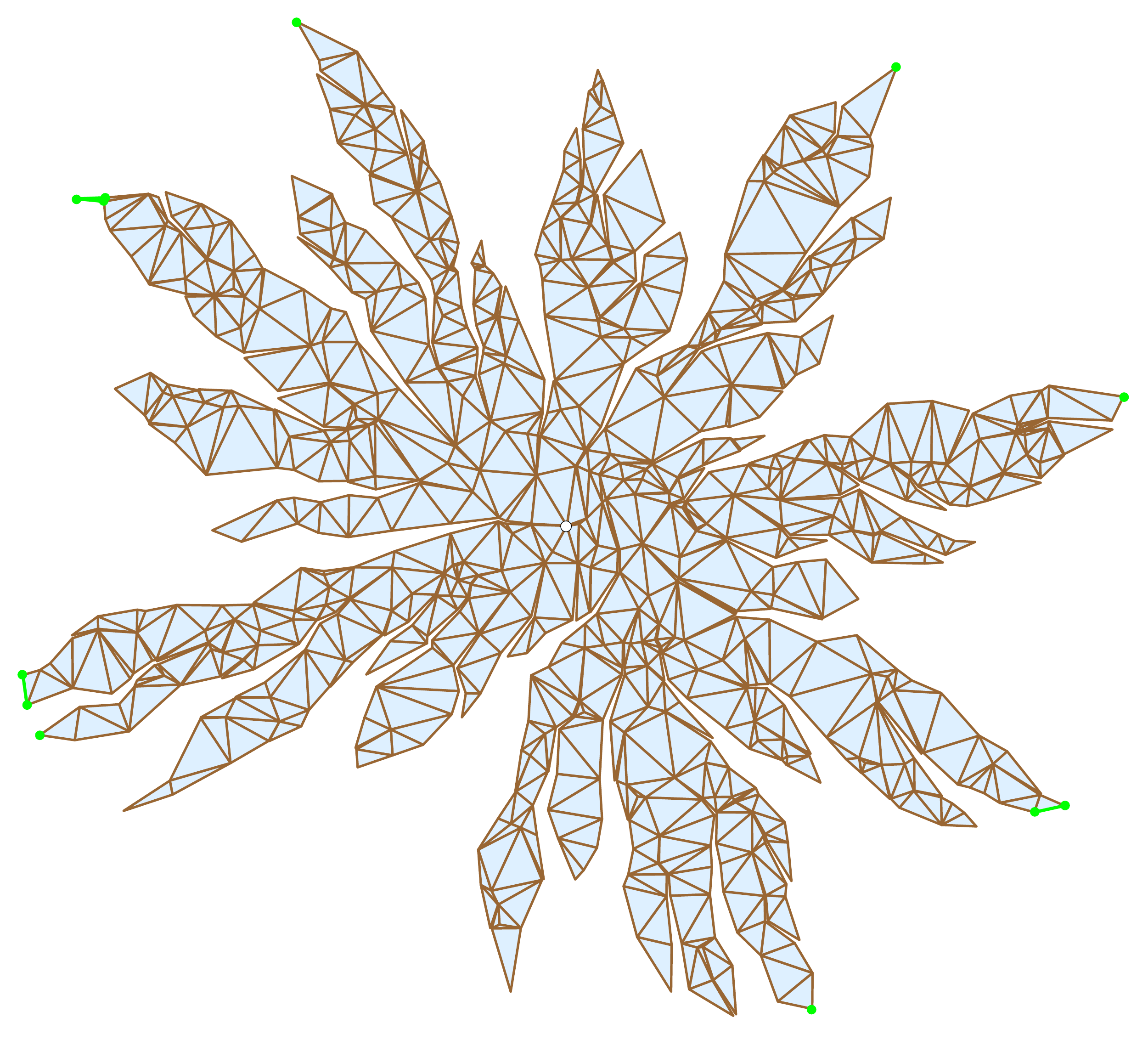}
\caption{Unfolding of a polyhedron of $n=500$ vertices.}
\figlab{Lay_s2_n500}
\end{figure}
\begin{figure}[htbp]
\centering
\includegraphics[width=1.0\linewidth]{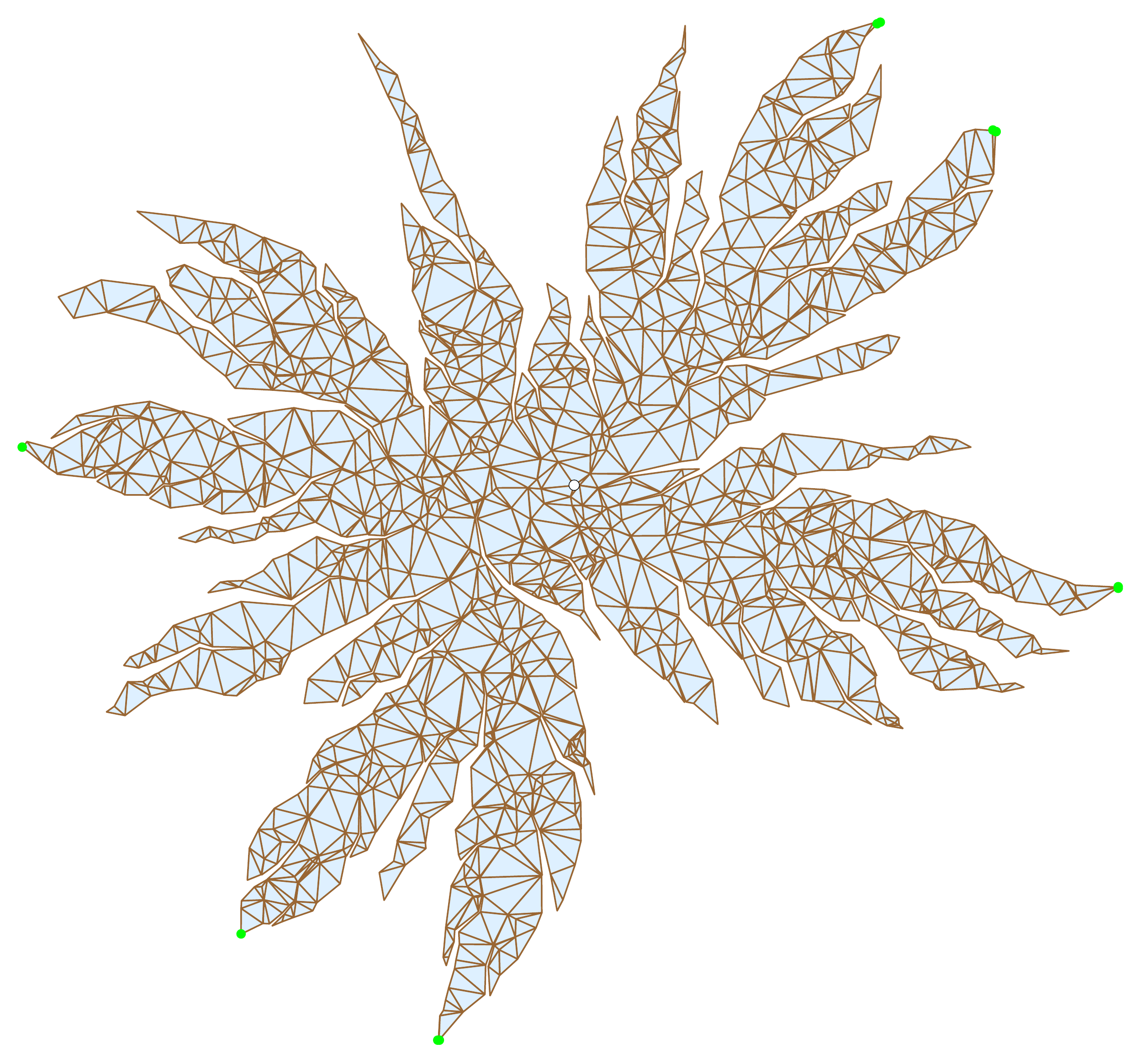}
\caption{Unfolding of a polyhedron of $n=1000$ vertices.}
\figlab{Lay_s2_n1000}
\end{figure}
\begin{figure}[htbp]
\centering
\includegraphics[width=1.0\linewidth]{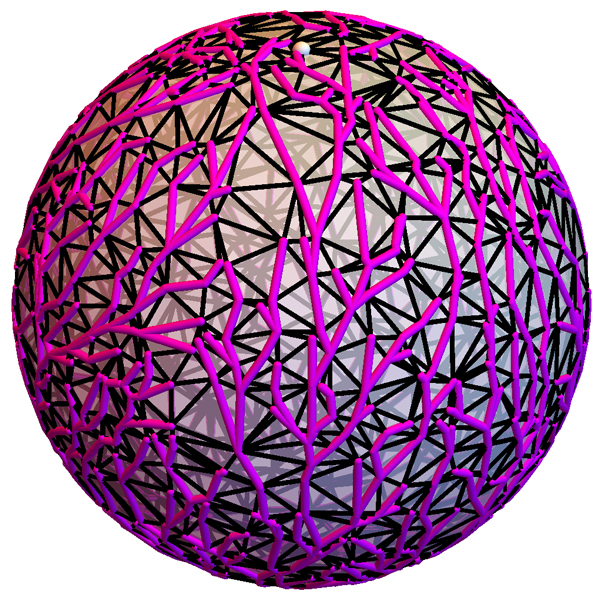}
\caption{Spanning radial monotone cut tree for polyhedron of $n=1500$ vertices.
North pole marked on top; tree root not visible at south pole.
Unfolding in Fig.~\protect\figref{Lay_s2_n1500}.}
\figlab{Cuts_s2_n1500}
\end{figure}
\begin{figure}[htbp]
\centering
\includegraphics[width=1.0\linewidth]{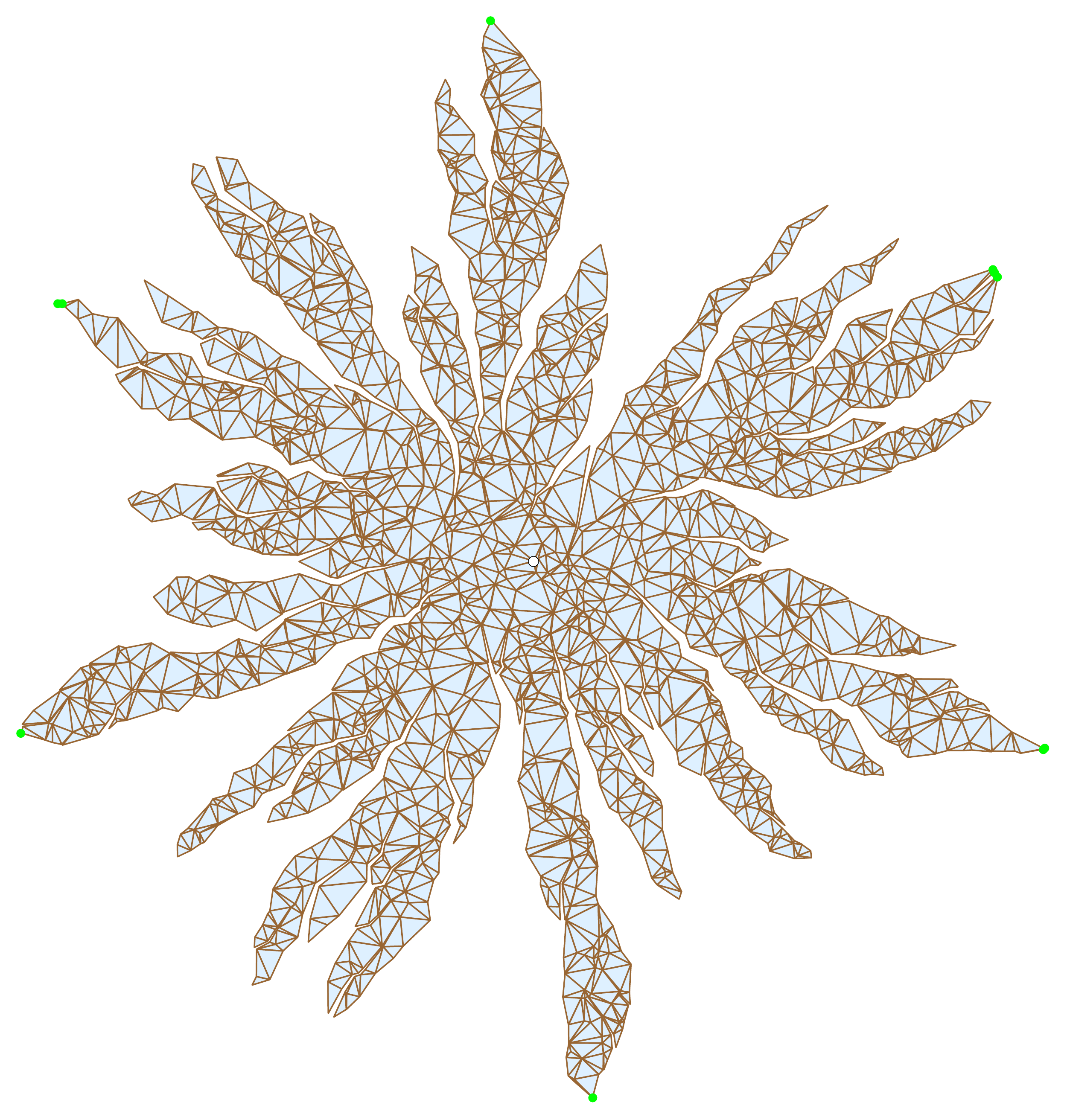}
\caption{Unfolding of polyhedron (Fig.~\protect\figref{Cuts_s2_n1500})
of $n=1500$ vertices.}
\figlab{Lay_s2_n1500}
\end{figure}

\subsection{Algorithm~2}
The algorithm is detailed below.

\begin{algorithm}[htbp]
\caption{Find rm cut forest $\F$ for polyhedron $\P$}
\DontPrintSemicolon

    \SetKwInOut{Input}{Input}
    \SetKwInOut{Output}{Output}

    \Input{Convex polyhedron $\P$ inscribed in sphere $\S$}
    \Output{(Usually) Radially monotone cut tree $\T$}
    
    Find ``bottommost'' triangle $\triangle abc$ (most southward normal). \;
    Set $\bC=(a,b,c,b)$. \;
    
    \BlankLine
    \tcp{Grow $\F$ from $\bC$ upward/inward}
    \BlankLine
    
    Sort interior vertices by geodesic distance on $\S$
    from north-pole $N$, those nearest $\bC$ first.
     
    \BlankLine
     
    \tcp{Grow $\F$:}
    
     $\F \leftarrow \varnothing$

     \ForEach{vertex $v_0$ in sorted (vertically ascending) order}{
   
     \ForEach{vertex $v_1$ already in $\F$ or on $\bC$}{
     
     \tcp{$v_1$ is below $v_0$.}
     
     If $v_0$ connects by a triangulation edge to $v_1$, set $e=(v_0,v_1)$.\;
     Check if the planar medial path $M(Q)$ from $v_0$ to $\bC$ in $\F+e$ is radially monotone.\;
     If so, record its worst turnangle $\t$ (with $\t>90^\circ$ not rm). \;
    
      }
       
      Choose the $e^*$ that has the best (minimum) $\t$. \;
      Or: report failure to find a radially monotone connection \& exit.
      
      $\F \leftarrow \F + e^*$
            
     }
     
     Return $\T \leftarrow \F + ab + bc$
     
\end{algorithm}

\subsection{Empirical Results}
\seclab{Empirical}
Algorithm~1 only provably works on planar convex domains when
they are non-obtusely triangulated, but the 3D examples we have so far
provided contain obtuse triangles: They were generated simply as convex
hulls of random points on a sphere $\S$,
partly because it is by no means straightforward to generate non-obtuse
triangulations, and partly because Algorithm~2 works with empirically
high frequency on spherical polyhedra.

My experiments show that Algorithm~2 finds
a radial monotone cut tree $\T$ for more than $95$\% of random spherical polyhedra
of $n=100$ vertices. For example, 
in one run, $20$ out of $1{,}000$ polyhedra
led to one (never more than one) vertex forced to select a non-rm connection.
In all cases, the ``cause" was an obtuse triangle, which if split,
led to a rm-tree and a non-overlapping unfolding.
Another run of $400$ random $200$-vertex spherical polyhedra
found $21$ single-vertex non-rm connections ($95$\%).

Fig.~\figref{Lay_s502_n100_nosplit} shows a typical example.
Fig.~\figref{Lay_s502_n100_nosplit}(a) shows the one vertex that makes a non-rm cut, in this case with turnangle $94^\circ$.
This angle in fact causes (barely visible) overlap in the unfolding.
Fig.~\figref{Lay_s502_n100_nosplit}(b) shows the same polyhedron with many obtuse angles split,
splitting the most common situation: one obtuse triangle adjacent to a non-obtuse
triangle. Now a different radially monotone cut tree is found, in particular,
resolving the problematic overlap.
\begin{figure}[htbp]
\centering
\includegraphics[width=1.0\linewidth]{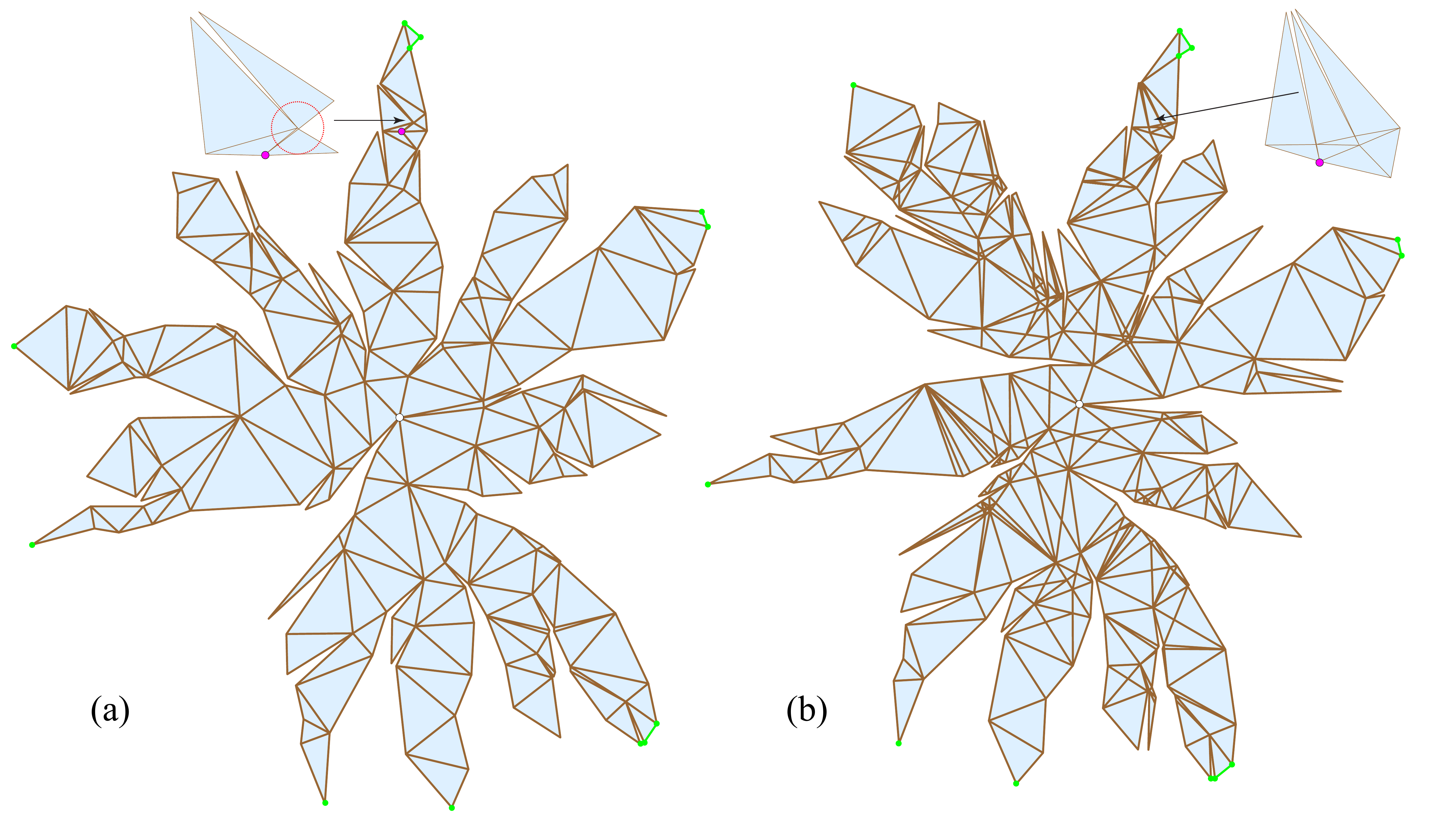}
\caption{(a)~$\P$ of $n=100$ vertices, with non-rm vertex cut marked.
(b)~The same $\P$ with many obtuse angles split, now of $n=203$ vertices.}
\figlab{Lay_s502_n100_nosplit}
\end{figure}

\paragraph{Caveats.}
\begin{sloppypar}
Although it is quite clear what constitutes a random spherical polyhedron---the convex
hull of random points on a sphere---it is less clear what is a random non-obtusely triangulated
spherical polyhedron. And in any case, 
the obtuse-splitting procedure I employed is ad hoc,
and creates ``vertices" of curvature $0$ not touching $\S$.
So I cannot make any justified empirical claims
concerning the performance of Algorithm~2 on such polyhedra.
\end{sloppypar}

Moreover, the $1{,}000$ examples mentioned in the Introduction were
achieved by (a)~running Algorithm~2 on spherical polyhedra, and
then, for each case where a rm tree was not found, (b)~rerunning it with obtuse triangles split as in Fig.~\figref{Lay_s502_n100_nosplit}.

However, I would like to emphasize that Algorithm~2 does not in any way ``search" for a
radially monotone cut tree: 
It uses whatever happens to be the bottommost triangle to form $\bC$,
and then it grows the spanning forest strictly in order of the geodesic
circle radii illustrated in Fig.~\figref{SphereLatitude}.
It never backtracks or considers alternatives (aside from choosing the
``best'' rm connection among those available below).
One could apply many heuristics to improve performance.
In fact, when I included several such heuristics--- selecting 
the ``most equilateral" triangle to become bottommost
(to maximally separate the three spanning tree roots of $\F$), 
and not following the concentric circle ordering, but rather growing $\F$
with the ``best" connection (in any direction) at each stage---the heuristic-laden algorithm
found rm cut trees with such high frequency I could not find a random counterexample.

\section{Questions \& Conjectures}

Algorithm~2 finds radially monotone cut trees for random spherical polyhedra with high frequency,
and then (because of Theorem~\thmref{LMR-nonint},
whose bounds apparently suffice for random polyhedra), unfolds them without
overlap.
Triangulations of spherical polyhedra are special: They are intimately
connected to Delaunay triangulations. This suggests:
\begin{question}
Does every planar Delaunay triangulation have a radially monotone spanning forest?
\end{question}
Note that the example in Appendix~1 (Fig.~\figref{RadialMonotoneCexTriangulation}) is non-Delaunay.
Also note that non-obtuse triangulations are automatically Delaunay.
The 3D analog is:
\begin{question}
Does every spherical (inscribed) polyhedron have a radially monotone cut tree?
\end{question}
My empirical explorations only suggest that this may hold for random spherical polyhedra.
\begin{question}
Does every non-obtusely triangulated convex polyhedron have a radially monotone cut tree?
\end{question}
I have so far not explored non-spherical polyhedra enough to form an opinion
on this question, but radial monotonicity seems intimately connected to
non-obtuseness.
\begin{question}
What is a natural definition of a
random, spherical (inscribed in a sphere), non-obtusely triangulated polyhedron?
And how could they be generated?
\end{question}
Without an answer to this question, the following conjecture is vague,
but nevertheless, I feel is justified:
\begin{conjecture}
A random spherical, non-obtusely triangulated polyhedron
has a radially monotone cut tree (as defined in Sec.~\secref{rm3D})
with high probability.
\end{conjecture}
More risky is the same conjecture without requiring the triangulation
to be non-obtuse:
\begin{conjecture}
A random spherical polyhedron
has a radially monotone cut tree
with high probability.
\end{conjecture}
One natural interpretation of ``high probability"
would be that, as the number of vertices $n \to \infty$, the probability goes to $1$.
This would essentially constitute an obverse of Fig.~\figref{CASJORoverlap},
which shows (empirically) that the probability of overlap
from a random spanning cut tree goes to $1$.

\paragraph{Addendum.}
Anna Lubiw\footnote{Personal communication, 27 July 2016.} 
informed me that my radially monotone paths are the
same as backwards ``self-approaching curves,"
introduced in~\cite{ikl-sac-99}, and explored in several papers since then.
Their definition for a curve $C$ is:
``for any three consecutive points $a,b,c$ in oriented order on $C$, the inequality
$d(a,c) \ge d(b,c)$ holds."
These curves have been studied for their length properties
and applications to graph drawing and routing in planar geometric graphs.
Some of the elementary properties I prove for radially monotone paths
were earlier derived in this literature.

\newpage
\section{Appendix~1}
\seclab{Appendix1}
\begin{figure}[htbp]
\centering
\includegraphics[width=1.0\linewidth]{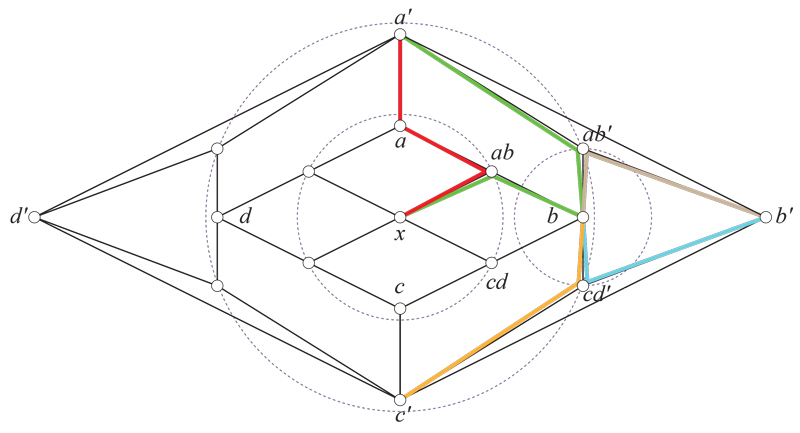}
\caption{A plane graph that cannot be spanned by an rm-forest. 
The circles indicate rm violations.}
\figlab{RadialMonotoneCexPlanarGraph}
\end{figure}
We first provide a counterexample for a plane graph, and then modify it
to achieve a triangulation.
\begin{lemma}
The plane graph $G$ shown in Fig.~\figref{RadialMonotoneCexPlanarGraph} has no rm-forest.
$G$ has strictly convex faces, and a convex boundary $\bC$.
\lemlab{Gplanar}
\end{lemma}
\begin{proof}
Let $\F$ be an rm-forest for $G$, and let $Q$ the unique path in a tree of $\F$
from $x$ to $\bG$.
\begin{enumerate}
\squeezelist
\item The four edges incident to $x$ are symmetric, so we choose $(x,ab)$ wlog.
There are two choices from $ab$.
\item $Q=(x, ab, a, a')$ is non-monotonic at $ab$ w.r.t. $x$: 
$(ab,a)$ cuts into the circle centered at $x$ passing through $ab$.
\item $(x, ab, b)$ is radially monotone.
\item Extending this path with the up edge $(b,ab')$ violates
radial monotonicty at $b$ w.r.t. $ab$.
So it must be instead extended downward, to 
$Q=(x, ab, b,cd')$.
\item $(x, ab, b,cd')$ cannot be extended to either $Q=(x, ab, b,cd',c')$ nor $Q=(x, ab, b,cd',b')$ monotonically:
the former violates radial monotonicity at $cd'$ w.r.t. $x$; the latter violates at $cd'$ w.r.t. $b$.
\end{enumerate}
Because of the symmetries of $G$, this exhausts all possible paths from $x$ to $\bG$.

The pentagon face $(a,ab,b,ab',a')$ and the quadrilateral face
$(b,ab',b',cd')$ are non-strictly convex faces,
but can be made strictly convex by slight movements of $ab$ and $b$ respectively,
without changing altering the monotonicity of any paths.
\end{proof}

\noindent
The claim in Lemma~\lemref{Gplanar} can be strengthened to a triangulation:

\begin{lemma}
The triangulated plane graph $G_T$ shown in Fig.~\figref{RadialMonotoneCexTriangulation} has no rm-forest.
\lemlab{GTriangulation}
\end{lemma}
\begin{proof}
\emph{Sketch.}
The structure of $G_T$ is based on that of $G$ in Lemma~\lemref{Gplanar}: 
$G_T$ is a triangulation of $G$.
Although there are now more paths from $x$ to $\bG_T$, they each still violate monotonicity,
either w.r.t. $x$, or w.r.t. $b$.
\end{proof}
\begin{figure}[htbp]
\centering
\includegraphics[width=1.0\linewidth]{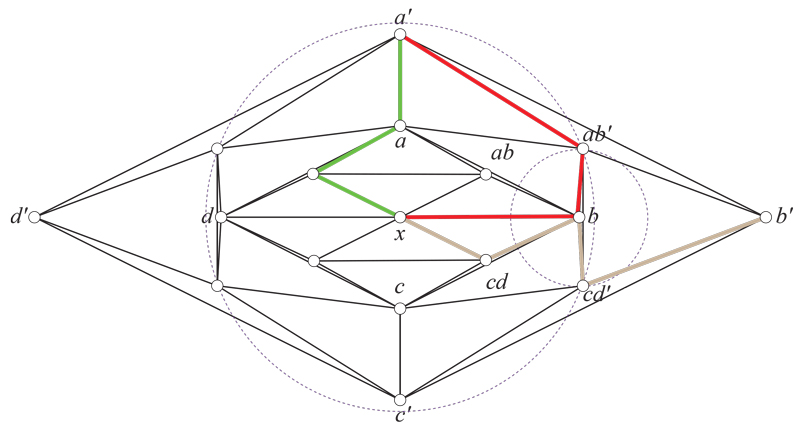}
\caption{A triangulated plane graph that cannot be spanned by an rm-forest.}
\figlab{RadialMonotoneCexTriangulation}
\end{figure}

\bibliographystyle{alpha}
\bibliography{geom}

\begin{thebibliography}{IKL99}

\bibitem[Bis16]{b-nonobtri-16}
Christopher Bishop.
\newblock Nonobtuse triangulations of {PSLG}s.
\newblock {\em Discrete Comput. Geom.}, 56(1):43--92, 2016.

\bibitem[DO07]{do-gfalop-07}
Erik~D. Demaine and Joseph O'Rourke.
\newblock {\em Geometric Folding Algorithms: Linkages, Origami, Polyhedra}.
\newblock Cambridge University Press, July 2007.
\newblock \url{http://www.gfalop.org}.

\bibitem[IKL99]{ikl-sac-99}
Christian Icking, Rolf Klein, and Elmar Langetepe.
\newblock Self-approaching curves.
\newblock {\em Math. Proc. Camb. Phil. Soc.}, 125:441--453, 1999.

\end{thebibliography}
\end{document}